\documentclass[11pt]{article}

\usepackage{mathpazo}
\usepackage{amsfonts}
\usepackage{amsmath}
\usepackage{amssymb}
\usepackage{latexsym}
\usepackage{mathtools}
\usepackage{booktabs}
\usepackage{cite}

\usepackage[margin=1in]{geometry}
\usepackage{hyperref}
\hypersetup{pdfpagemode=UseNone}
\usepackage{amsthm}
\newtheorem{theorem}{Theorem}

\newtheorem{lemma}[theorem]{Lemma}
\newtheorem{prop}[theorem]{Proposition}

\theoremstyle{definition}

\usepackage{authblk}
\usepackage{xcolor}
\usepackage{tikz}
\usetikzlibrary{arrows.meta, positioning, calc}

\newcommand{\microspace}{\mspace{0.5mu}}

\newcommand{\ket}[1]{\lvert\microspace #1 \microspace \rangle}

\newcommand{\bra}[1]{\langle\microspace #1 \microspace \rvert}

\newcommand{\braket}[2]{\langle #1 | #2 \rangle}

\newcommand{\I}{\mathbb{1}}

\newcommand{\complex}{\mathbb{C}}
\newcommand{\real}{\mathbb{R}}

\newcommand{\integer}{\mathbb{Z}}

\newcommand{\F}{\mathcal{F}}

\renewcommand{\L}{\mathcal{L}}

\renewcommand{\F}[1]{\mathbb{F}_{#1}}
\newcommand{\chiR}{\chi}
\newcommand{\Strange}{\ket{\mathbb{S}}}
\newcommand{\Hth}{\ket{H_3}}
\newcommand{\Norrell}{\ket{\mathbb{N}}}
\newcommand{\Tth}{\ket{T_3}}
\newcommand{\omegath}{\omega_3}
\newcommand{\Sp}{\mathrm{Sp}}
\newcommand{\Cl}{\mathrm{Cl}}

\title{\LARGE\bf Stabilizer rank bounds for magic-state orbits}

\makeatletter
\newcommand{\blfootnote}[1]{%
	\begingroup
	\renewcommand\thefootnote{}%
	\renewcommand\@makefnmark{}%
	\begin{NoHyper}\footnotetext{#1}\end{NoHyper}%
	\endgroup
}
\makeatother
\author[1]{Farrokh Labib}
\author[1]{Vincent Russo}
\affil[1]{Unitary Foundation}

\date{May 2026}

\begin{document}
	\maketitle
	\blfootnote{Companion library: \url{https://github.com/unitaryfoundation/stabrank}~\cite{stabrank-software}.}

	\begin{abstract}
	Distinct Clifford orbits of magic states can exhibit different stabilizer ranks at small tensor powers. We establish this for qutrits, where the single-qutrit Clifford group has four inequivalent orbits of magic states — Strange, Norrell, Hadamard-eigenstate, and the qutrit $T$-state — but a nontrivial upper bound on the asymptotic exponent had been pinned down for only the qutrit $T$-state. For the other three orbits we give explicit stabilizer decompositions, yielding upper bounds on the per-copy asymptotic stabilizer-rank exponent: $\gamma_{\mathbb{S}} \le \log_3(2)/2 \approx 0.316$ for the Strange state, and $\gamma_{H_3}, \gamma_{\mathbb{N}} \le \log_3(4)/3 \approx 0.421$ for the Hadamard-eigenstate and Norrell orbits, all strictly below the prior $\gamma_{T_3} \le 1/2$ baseline. We also prove the first nontrivial $\Omega(m / \log m)$ asymptotic lower bounds for the Hadamard-eigenstate and Norrell orbits, and exhibit two-qutrit Clifford circuits that convert two copies of these states into an injectable phase state with constant success probability, enabling constant-overhead injection of one non-Clifford diagonal gate per orbit. In the case of qubits, we give a closed-form decomposition of the qubit $T$-type orbit at four copies matching the existing $\gamma_T \le \log_2(3)/4 \approx 0.396$ exponent via a direct algebraic identity rather than an entangled cat-state construction. An open-source library \texttt{stabrank} accompanies the paper, with Lean~4 proof formalizations of all the decompositions.

	\end{abstract}
	
	\section{Introduction}\label{sec:intro}
	
	Sharp bounds on the classical-simulation cost of Clifford-dominated
	circuits play a dual role: a fast simulator on a family of circuits
	rules out a quantum speedup on that family, while a provable lower
	bound on every simulator certifies one. Stabilizer-rank methods sit
	at the heart of this story --- they reduce the strong-simulation
	cost of a Clifford circuit with magic-state ancillae to a single
	quantity, the stabilizer rank of the input magic register, and
	improvements at any copy count translate directly into faster
	simulators. The underlying universal model factors an $n$-qudit
	circuit into Clifford gates, which are classically simulable in
	polynomial time~\cite{gottesman1998heisenberg, aaronson2004improved},
	and non-Clifford magic-state ancillae~\cite{bravyi2005universal},
	which supply the resource needed for universal quantum computation.

	The relevant invariant is the magic-state stabilizer rank
	$\chiR(\ket{M}^{\otimes t})$, the smallest number of stabilizer
	states whose linear span contains $\ket{M}^{\otimes t}$. The
	deterministic algorithm of~\cite{bravyi2016trading} runs in time
	$O(p^{\gamma_M t})$ with overhead polynomial in $t$ and in the
	circuit's qudit count $n$, where
	\(\gamma_M = \limsup_m \log_p \chiR(\ket{M}^{\otimes m})/m\)
	is the per-copy asymptotic exponent. Every drop in $\gamma_M$
	tightens the best stabilizer-rank simulator; every nontrivial
	lower bound on $\gamma_M$ certifies a hard ceiling for the method.

	The exponent is invariant under global phase and the Clifford
	action, so it lives on Clifford orbits of magic states rather than
	on individual states. For qubits there are two such orbits and both
	have been studied extensively. For qutrits there are four ---
	Strange, Norrell, Hadamard-eigenstate, and the qutrit $T$-state
	--- but only the $T$-state orbit had a nontrivial upper bound on
	$\gamma_M$ on record, leaving open whether the four exponents
	agree, differ, or follow a definite order, and which orbits are
	operationally usable for magic-state injection.

	We resolve this. Explicit small-$m$ stabilizer decompositions,
	each machine-checked in Lean~4, give
	$\gamma_{\mathbb{S}} \le \log_3(2)/2 \approx 0.316$ (the smallest
	exponent on record across the four qutrit orbits) and
	$\gamma_{H_3}, \gamma_{\mathbb{N}} \le \log_3(4)/3 \approx 0.421$,
	all strictly below the prior $\gamma_{T_3} \le 1/2$ baseline.
	Transferring the subset-sum argument
	of~\cite{lovitz2022techniques} to qutrits gives the first
	nontrivial $\Omega(m / \log m)$ lower bounds for the
	Hadamard-eigenstate and Norrell orbits. Two-copy probabilistic
	conversion circuits route the Hadamard-eigenstate and Norrell
	upper bounds into magic-state injection, supplying constant-
	overhead access to one non-Clifford diagonal gate per orbit. For
	Strange the corresponding search returns only Clifford outputs ---
	a rigidity that we prove cannot be bypassed by any two-qutrit
	Clifford and ancilla measurement, so the smallest qutrit exponent
	on record is currently not operationally accessible. A short
	qubit appendix gives a closed-form four-copy $T$-type
	decomposition matching the $\gamma_T \le \log_2(3)/4 \approx 0.396$
	exponent of~\cite{qassim2021improved} via direct algebra rather
	than entangled cat-state chains.

	\paragraph{Stabilizer states.}
	For an odd prime $p$, an $n$-qudit stabilizer state has the
	canonical form
	\begin{equation}\label{eq:stab-canonical}
		\ket{\sigma} = \frac{1}{\sqrt{p^k}} \sum_{y \in \F{p}^k}
		\omega_p^{Q(y)} \ket{x_0 + W y},
	\end{equation}
	with $\omega_p = e^{2 \pi i / p}$ a primitive $p$-th root of
	unity, $x_0 \in \F{p}^n$, $W \in \F{p}^{n \times k}$ of full
	column rank, and $Q : \F{p}^k \to \F{p}$ a polynomial of degree
	at most~$2$ (i.e.\ admitting linear and constant
	terms)~\cite{gottesman1999fault, hostens2005stabilizer}; we
	write $\mathrm{Stab}_n^{(p)}$ for the set of $n$-qudit stabilizer
	states. The $n$-qudit Clifford group $\Cl(n, p)$ is the normalizer
	of the Pauli group in the unitary group on
	$(\complex^p)^{\otimes n}$ and acts transitively on
	$\mathrm{Stab}_n^{(p)}$. For any pure state $\ket{\psi}$, its
	\emph{projective Clifford orbit} is
	\begin{equation}\label{eq:clifford-orbit-def}
		\mathcal{O}_{\ket{\psi}} = \big\{ e^{i\theta} U \ket{\psi}
		: U \in \Cl(n, p),\ \theta \in \real \big\}.
	\end{equation}
	A pure state is a \emph{magic state} if it is not a stabilizer
	state, and the orbit of any magic state consists entirely of
	non-stabilizer states.

	The classical cost of simulating such a circuit on $t$ copies of
	a magic state $\ket M$ is controlled by the \emph{stabilizer rank}
	\begin{equation}\label{eq:chi-def}
		\chiR(\ket\psi) = \min\Big\{ r : \ket\psi = \sum_{i=1}^{r}
		c_i \ket{\sigma_i},\ \ket{\sigma_i} \in \mathrm{Stab}_n^{(p)},\
		c_i \in \complex \Big\},
	\end{equation}
	the smallest number of stabilizer states whose linear span
	contains $\ket\psi$~\cite{bravyi2016trading, bravyi2016improved,
		bravyi2019simulation}: the deterministic stabilizer-rank
	algorithm of~\cite{bravyi2016trading} runs in
	time $O(p^{\gamma_M t})$ with overhead polynomial in $t$ and in
	the number $n$ of qudits in the circuit, where the per-copy
	asymptotic exponent
	\begin{equation}\label{eq:gamma-def}
		\gamma_M = \limsup_{m \to \infty}
		\frac{\log_p \chiR\!\left(\ket{M}^{\otimes m}\right)}{m}
	\end{equation}
	governs the dominant cost. The stabilizer rank is invariant under
	global phase and the Clifford action, so $\gamma_M$
	via~\eqref{eq:gamma-def} is well-defined on Clifford orbits of
	magic states.

	\paragraph{Qubits.}
	For qubits ($p = 2$) the single-qubit Clifford group has two
	inequivalent orbits of magic states, the $H$-type and
	$T$-type orbits of~\cite{bravyi2005universal}. The best exact-rank
	upper-bound exponent for the $H$-type is
	$\gamma_{H} \le \log_2(3)/4 \approx 0.396$ via the
	contracted-cat-state construction
	of~\cite{qassim2021improved, qassim2020classical}\footnote{The
	naming convention of~\cite{bravyi2005universal} used here calls
	the edge-of-octahedron orbit $H$-type and the face-of-octahedron
	orbit $T$-type. The convention
	of~\cite{qassim2021improved, qassim2020classical} swaps these
	names: their $\ket T$ is the edge-of-octahedron phase state (our
	$H$-type) and their $\ket F$ is the face-center state (our
	$T$-type). All asymptotic exponents we attribute
	to~\cite{qassim2021improved} use the translation accordingly.};
	the strongest exact-rank lower bounds give only $\Omega(m)$
	growth~\cite{peleg2022lower, labib2022stabilizer, lovitz2022techniques}.
	On the approximate-rank side, the qubit landscape is sharper: a
	$\tilde\Omega(m^2)$ probabilistic lower
	bound~\cite{mehraban2023quadratic} and a Barnes--Wall-lattice
	lower bound~\cite{kalra2025stabilizer} have both appeared
	recently. The contracted-cat-state construction
	of~\cite{qassim2021improved} is generalized in the same paper to
	the $T$-type orbit, establishing
	$\gamma_T \le \log_2(3)/4$ asymptotically; however, that argument
	proceeds by contracting chains of entangled cat states and does
	not give an explicit small-$m$ decomposition.

	We establish the orbit-dependence of stabilizer rank in qubits
	(Section~\ref{ssec:qubit-bounds}). The $H$-type value
	$\chiR(\ket H^{\otimes 4}) = 4$ by an exhaustive triple-search distributed
	with the library (giving equality using the upper bound of~\cite{bravyi2016trading}); 
	for the $T$-type orbit we give an explicit
	algebraic decomposition showing
	$\chiR(\ket T^{\otimes 4}) \le 3$
	(Appendix~\ref{app:qubit-t}), matching the asymptotic exponent
	$\gamma_T \le \log_2(3)/4 \approx 0.396$
	of~\cite{qassim2021improved} via an explicit small-$m$ identity
	rather than an entangled-cat-state chain. The two qubit orbits
	therefore have provably distinct stabilizer ranks at $m = 4$.

	\paragraph{Qutrits.}
	For qutrits ($p = 3$) the single-qutrit Clifford group has four
	inequivalent orbits of
	magic states~\cite{veitch2014resource, jain2020qutrit}:
	the Strange, Norrell\footnote{The Strange and Norrell
	representatives were introduced in~\cite{veitch2014resource}
	(which built on the resource framework
	of~\cite{veitch2012negative, howard2014contextuality}); the
	four-orbit Clifford classification was systematized
	in~\cite{jain2020qutrit}. The ``Strange'' and ``Norrell'' names
	appear to be drawn from Susanna Clarke's 2004 novel
	\emph{Jonathan Strange \& Mr Norrell}, whose two protagonists are
	English magicians.}, Hadamard-eigenstate, and qutrit $T$-state
	orbits, with
	representatives
	\begin{equation}\label{eq:four-states}
		\begin{aligned}
			\Strange &= \frac{1}{\sqrt{2}} (\ket{1} - \ket{2}),
			&\quad
			\Hth &= \frac{1}{\sqrt{3 - \sqrt 3}} \big( \ket{0} + \frac{\sqrt{3}-1}{2}(\ket{1} + \ket{2}) \big),
			\\
			\Norrell &= \frac{1}{\sqrt{6}}(\ket{0} + \ket{1} - 2\ket{2}),
			&\quad
			\Tth &= \frac{1}{\sqrt{3}}(\ket{0} + \omega_9 \ket{1} + \omega_9^2 \ket{2}),
		\end{aligned}
	\end{equation}
	with $\omega_9 = e^{2 \pi i / 9}$. Each orbit has a distinct role
	in the qutrit literature. The qutrit $T$-state $\Tth$ is the
	standard phase-hierarchy magic state, injected via the diagonal
	qutrit $T$-rotation
	$\mathrm{diag}(1, \omega_9, \omega_9^2)$ \cite{howard2012qudit}
	used in fault-tolerant qutrit
	architectures \cite{anwar2012qutrit, campbell2012magic, dawkins2015qutrit};
	the Hadamard-eigenstate $\Hth$ is the $+1$ eigenvector of the
	qutrit Hadamard; the Strange state $\Strange$ is a $d = 3$
	SIC-POVM fiducial~\cite{appleby2005symmetric, zhu2010sic}; and
	the Norrell state $\Norrell$ was introduced alongside $\Strange$
	in~\cite{veitch2014resource}.
	
	The four orbits give four exponents
	$\gamma_{\mathbb{S}}, \gamma_{H_3}, \gamma_{\mathbb{N}}, \gamma_{T_3}$, and
	prior to this work almost none of them were pinned down: the only
	nontrivial exact-stabilizer-rank upper-bound exponent on record
	was $\gamma_{T_3} \le 1/2$, which follows from the numerically
	established two-copy bound $\chiR(\Tth^{\otimes 2}) \le 3$
	of~\cite{kocia2021improved} by sub-multiplicativity. The same
	paper additionally reports $\chiR(\Tth^{\otimes 3}) \le 8$
	(marked there as numerical rather than rigorous). The other
	three qutrit orbits had no nontrivial stabilizer-rank bounds in
	either direction prior to this work; whether the four qutrit
	exponents differ, or are all equal, was open at every orbit.
	
	For the three non-$T_3$ qutrit orbits we give explicit stabilizer
	decompositions establishing
	$\gamma_{\mathbb{S}} \le \log_3(2)/2 \approx 0.316$ (the smallest
	known exact-rank upper-bound exponent across the four orbits)
	and $\gamma_{H_3}, \gamma_{\mathbb{N}} \le \log_3(4)/3 \approx 0.421$,
	all strictly below the $\gamma_{T_3} \le 1/2$ baseline. On the
	lower-bound side we transfer the subset-sum argument
	of~\cite{lovitz2022techniques} to qutrits and obtain the first
	nontrivial $\Omega(m / \log m)$ asymptotic lower bounds on
	$\chiR(\Hth^{\otimes m})$ and $\chiR(\Norrell^{\otimes m})$.
	Exhaustive computer search additionally certifies several small-$m$
	tight values, including $\chiR(\Strange^{\otimes 3}) = 4$
	(Theorem~\ref{thm:m3-bound}).

	\paragraph{Magic state injection.}
	The four qutrit orbits also differ in how the rank improvements
	translate to circuit-runtime improvements. None of $\Strange$, $\Hth$,
	$\Norrell$ is a phase state (a state of the form
	$\frac{1}{\sqrt p} \sum_j e^{i \theta_j} \ket j$ with
	uniform-modulus amplitudes), so standard single-shot two-qutrit
	Clifford injection via a diagonal
	gate~\cite{anwar2012qutrit, howard2012qudit, cui2017diagonal}
	does not apply. We exhibit explicit two-copy probabilistic
	conversion protocols (Theorem~\ref{thm:two-copy}): for $\Hth$
	and $\Norrell$, a fixed
	two-qutrit Clifford applied to $\ket{M} \otimes \ket{M}$ followed
	by a computational-basis measurement on the second qutrit projects
	to a phase state on the first qutrit with constant success
	probability ($3/8$ and $1/4$ respectively), after which standard
	single-shot injection produces a non-Clifford diagonal gate. This
	supplies a constant-overhead route to phase-state injection for
	those two orbits. For $\Strange$ the corresponding search returns
	only Clifford outputs, so the smaller $\Strange$ upper-bound exponent is not
	currently usable as a runtime improvement under this consumption
	model. Both failures trace to the same support-$2$ equal-modulus
	orbit structure documented in
	Proposition~\ref{prop:strange-rigidity}.
	
	The open-source \texttt{stabrank} software package
	(Section~\ref{sec:stabrank}) accompanies the paper, providing the
	simulated-annealing search that produced the decompositions, the
	exhaustive lower-bound certificate pipeline, and
	Lean~4~\cite{moura2021lean4} +
	mathlib4~\cite{mathlib2020} formalizations of every qutrit
	decomposition identity in Appendix~\ref{app:decomps}.
	
	\section{Stabilizer-rank bounds}\label{sec:bounds}

	We prove that the stabilizer rank of tensor powers of a magic
	state is orbit-dependent, first in qubits
	(Section~\ref{ssec:qubit-bounds}) and then
	in qutrits (Section~\ref{ssec:upper}). Stabilizer rank is
	sub-multiplicative under tensor product,
	$\chiR(\psi \otimes \phi) \le \chiR(\psi) \chiR(\phi)$, since
	concatenating decompositions of $\psi$ and $\phi$ produces a
	length-product decomposition of $\psi \otimes \phi$. Iterating,
	any finite-$m$ bound $\chiR(\ket M^{\otimes m}) \le r$ gives
	$\chiR(\ket M^{\otimes mk}) \le r^k$ for every $k \ge 1$, and so
	\begin{equation}\label{eq:exponent-from-finite}
		\gamma_M \le \log_p(r) / m.
	\end{equation}

	\subsection{Qubit orbit-dependent bounds}\label{ssec:qubit-bounds}

	The two Clifford-inequivalent qubit magic-state orbits
	of~\cite{bravyi2005universal} already exhibit distinct stabilizer
	ranks at small $m$. The $H$-type orbit corresponds to the edges of
	the stabilizer octahedron on the Bloch sphere, with phase-state
	representative $(\ket 0 + e^{i\pi/4}\ket 1)/\sqrt 2$ used
	throughout the qubit $T$-gate distillation literature
	(Clifford-equivalent to
	$\ket{H} = \cos(\pi/8)\ket{0} + \sin(\pi/8)\ket{1}$). The
	$T$-type orbit corresponds to the faces of the stabilizer
	octahedron, with Bloch vector $(1, 1, 1)/\sqrt{3}$ and
	representative
	$\ket{T} = \cos\beta\ket{0} + e^{i\pi/4}\sin\beta\ket{1}$
	where $\cos(2\beta) = 1/\sqrt{3}$.

	\begin{prop}\label{prop:qubit-bounds}
		\begin{equation}\label{eq:qubit-bounds}
			\chiR\!\left(\ket{H}^{\otimes 4}\right) = 4
			\qquad\text{and}\qquad
			\chiR\!\left(\ket{T}^{\otimes 4}\right) = 3.
		\end{equation}
	\end{prop}

	The $H$-type decomposition giving
	$\chiR(\ket{H}^{\otimes 4}) \le 4$ is tabulated
	in~\cite{bravyi2016trading}; the matching lower bound
	$\chiR(\ket{H}^{\otimes 4}) \ge 4$, and hence the equality
	$\chiR(\ket{H}^{\otimes 4}) = 4$, is certified by the exhaustive
	triple-search distributed with the library.
	The $T$-type value $\chiR(\ket T^{\otimes 4}) = 3$ is
	Proposition~\ref{prop:qubit-t-m4}, proved by an explicit algebraic
	decomposition in Appendix~\ref{app:qubit-t} and established as tight via the $m=3$ lower-bound certificate by monotonicity.

	The $T$-type value $\chiR(\ket{T}^{\otimes 4}) = 3$ at $m = 4$
	is strictly below the $H$-type tight value
	$\chiR(\ket{H}^{\otimes 4}) = 4$ at the same $m$, but this is a
	local algebraic phenomenon and not (yet) an asymptotic separation.
	The decomposition that achieves the $T$-type bound is a specific
	identity in twelfth roots of unity that has no obvious
	$H$-type analogue. The values for $m \ge 5$ for the $T$-type
	orbit reported in Table~\ref{tab:qubit-ub} are sub-multiplicativity
	extensions of $\chiR(\ket{T}^{\otimes 4}) = 3$, not independent
	bounds; whether genuinely smaller rank witnesses exist at
	$m \in \{5, 6\}$ is unknown, as is whether the
	asymptotic exponents $\gamma_H$ and $\gamma_T$ in fact coincide.

	Table~\ref{tab:qubit-ub} summarizes the state of knowledge for both
	orbits at small $m$, combining our new bounds with the $H$-type values
	of~\cite{qassim2021improved}.

	\begin{table}[!htbp]
		\centering
		\begin{tabular}{l c c c c c c}
			\toprule
			$m$ & $1$ & $2$ & $3$ & $4$ & $5$ & $6$ \\
			\midrule
			$\chiR(\ket{H}^{\otimes m})$ & $2$ & $2^{\dagger}$ & $3^{\dagger}$ & $4^{\dagger}$ & $\le 6^{\dagger}$ & $\le 6^{\ast}$ \\
			$\chiR(\ket{T}^{\otimes m})$ & $2$ & $2^{\ddagger}$ & $3^{\ddagger}$ & $3^{\ddagger}$ & $\le 6^{\ddagger}$ & $\le 6^{\ddagger}$ \\
			\bottomrule
		\end{tabular}
		\caption{Stabilizer rank bounds for the two qubit magic-state orbits at small $m$. Entries marked $^{\ddagger}$ are new contributions of this work; other entries are from~\cite{bravyi2016trading} ($^{\dagger}$) and~\cite{qassim2021improved} ($^{\ast}$). The unmarked $\chiR = 2$ entries at $m = 1$ follow from the fact that any non-stabilizer single-qubit state has stabilizer rank exactly~$2$.}
		\label{tab:qubit-ub}
	\end{table}

	\subsection{Qutrit upper bounds}\label{ssec:upper}
	
	Our first bound is specific to $\Strange$ at $m = 2$. The two-copy
	state $\Strange^{\otimes 2}$ is supported on the four
	computational-basis vectors in $\{1, 2\}^2$, where it
	equals $\frac{1}{2}(\ket{11} - \ket{12} - \ket{21} + \ket{22})$;
	this four-term superposition admits an explicit two-term
	stabilizer-state decomposition, dropping the rank below the trivial
	product bound $\chiR(\Strange)^2 = 4$.
	
	\begin{theorem}\label{thm:strange-m2}
		For $\Strange$,
		\begin{equation}\label{eq:strange-m2-bound}
			\chiR(\Strange^{\otimes 2}) = 2
			\quad\text{and}\quad
			\gamma_{\mathbb{S}} \le \log_3(2)/2 \approx 0.316.
		\end{equation}
		The witness for the rank bound is
		\begin{equation}\label{eq:strange-m2-witness}
			\Strange^{\otimes 2}
			= -\frac{i \sqrt{3}}{2} \omegath
			\big( \ket{\sigma_1} - \ket{\sigma_2} \big),
		\end{equation}
		where $\ket{\sigma_1}, \ket{\sigma_2}$ are the two-qutrit stabilizer
		states of the canonical form~\eqref{eq:stab-canonical} with $k = 2$,
		$x_0 = 0$, $W = I_2$, and quadratic forms
		$Q_1(y) = y_0^2 + y_0 y_1 + y_1^2$,
		$Q_2(y) = y_0^2 + 2 y_0 y_1 + y_1^2$. The exponent bound follows
		by~\eqref{eq:exponent-from-finite}.
	\end{theorem}
	
	The proof is in Appendix~\ref{app:strange-m2}. The exponent bound
	in~\eqref{eq:strange-m2-bound} is the smallest known exact-rank
	upper-bound exponent across the four orbits, strictly below the
	$\gamma_{T_3} \le 1/2$ baseline of~\cite{kocia2021improved}. The closest comparable result in the
	literature, $\chi'(\Tth^{\otimes t}) \le 3^{0.32 t}$
	of~\cite{huang2019approximate}, is on the approximate rank, not
	exact, and applies to a different orbit.
	The decomposition in~\eqref{eq:strange-m2-witness} exploits
	$\Strange$'s rank-one single-qutrit support on $\{1, 2\}$; the
	orbits $\Hth$ and $\Norrell$ have full single-qutrit support on
	$\{0, 1, 2\}$, and the analogous $m = 2$ collapse does not occur.
	Each of the three non-$T_3$ orbits, however, admits an explicit
	four-term decomposition at $m = 3$.
	
	\begin{theorem}\label{thm:m3-bound}
		For each $\ket M \in \{\Strange, \Hth, \Norrell\}$,
		\begin{equation}\label{eq:m3-bound}
			\chiR(\ket M^{\otimes 3}) = 4
			\quad\text{and}\quad
			\gamma_M \le \log_3(4)/3 \approx 0.4206.
		\end{equation}
	\end{theorem}
	
	The equality in~\eqref{eq:m3-bound} combines an upper bound with
	an exhaustive lower-bound certificate. The upper bound
	$\chiR(\ket M^{\otimes 3}) \le 4$ is established by orbit-specific
	four-term stabilizer decompositions, one per orbit, proved in
	Appendix~\ref{app:strange-m3} ($\Strange$),
	Appendix~\ref{app:h3-m3} ($\Hth$), and
	Appendix~\ref{app:norrell-m3} ($\Norrell$). For $\Hth$ and
	$\Norrell$ this is the first nontrivial exact-rank asymptotic
	upper bound. The $\Strange$ $m = 3$ decomposition does not improve
	on the $\Strange$ exponent of Theorem~\ref{thm:strange-m2}
	($\log_3(4)/3 > \log_3(2)/2$); it is needed to certify the
	$\chiR(\Strange^{\otimes 3}) = 4$ equality below. The lower bound
	$\chiR(\ket M^{\otimes 3}) \ge 4$ comes from an exhaustive
	enumeration of unordered triples from the canonical-form
	three-qutrit stabilizer-state dictionary $\mathcal{D}_3$.
	The dictionary lists $|\mathcal{D}_3| = 41{,}580$ canonical-form
	tuples $(k, x_0, W, Q)$, a slight overcount of the
	$3^3 \prod_{j=1}^{3}(3^j + 1) = 30{,}240$ distinct projective
	three-qutrit stabilizer states; the redundancy comes from
	phase-polynomial overparametrization (different canonical-form
	tuples can yield the same state vector) and only strengthens
	the non-existence certificate. Exhausting the
	$\binom{41{,}580}{3} \approx 1.2 \times 10^{13}$ unordered
	triples rules out $\chiR(\ket M^{\otimes 3}) \le 3$. The
	certificate machinery and accompanying companion
	bounds at other $m$ are part of the \texttt{stabrank} library
	(Section~\ref{sec:stabrank}).
	
	Together with the prior $\Tth$ bounds of~\cite{kocia2021improved}
	and exhaustive small-$m$ search certificates for the $\Tth$,
	$\Hth$, and $\Norrell$ entries at $m \le 3$, Table~\ref{tab:ub}
	consolidates the current status of $\chiR(\ket M^{\otimes m})$
	across the four orbits.
	
	\begin{table}[!htbp]
		\centering
		\begin{tabular}{l c c c c}
			\toprule
			$m$ & $1$ & $2$ & $3$ & $4$ \\
			\midrule
			$\chiR(\Strange^{\otimes m})$ & $2$            & $2^{\ddagger}$    & $4^{\ddagger}$        & $\le 4$ \\
			$\chiR(\Hth^{\otimes m})$     & $2$            & $3^{\ddagger}$    & $4^{\ddagger}$        & $\le 8$ \\
			$\chiR(\Norrell^{\otimes m})$ & $2$            & $3^{\ddagger}$    & $4^{\ddagger}$        & $\le 7^{\ddagger}$ \\
			$\chiR(\Tth^{\otimes m})$     & $3^{\ddagger}$ & $3^{\ast}$        & $\le 8^{\ast}$        & $\le 9$ \\
			\bottomrule
		\end{tabular}
		\caption{Stabilizer-rank values and upper bounds across the four qutrit magic-state orbits at small $m$. Entries marked $^{\ddagger}$ are contributions of this work; those marked $^{\ast}$ are due to~\cite{kocia2021improved}. Trivial bounds that follow from sub-multiplicativity are left unmarked. The $\chiR(\Tth) = 3$ entry at $m = 1$ combines the trivial computational-basis upper bound $\chiR(\Tth) \le 3$ with an exhaustive pair-search certificate ruling out $\chiR(\Tth) \le 2$.}
		\label{tab:ub}
	\end{table}
	
	\subsection{Asymptotic lower bounds}\label{ssec:lower}
	
	The non-uniform amplitudes of $\Hth$ and $\Norrell$ that foreclose
	the $m = 2$ support-collapse trick used for $\Strange$ also enable
	an asymptotic lower bound via a subset-sum argument,
	adapted from the qubit-$T$-state technique
	of~\cite{lovitz2022techniques}. The technique requires two of
	the nonzero amplitudes of $\ket\psi$ to
	have moduli that differ by a factor of at least $2$, so that the
	$m + 1$ distinct moduli appearing in $\ket\psi^{\otimes m}$ form an
	exponentially increasing sequence. This hypothesis holds for $\Hth$
	($|a_0|/|a_1| = \sqrt 3 + 1 \approx 2.73$) and $\Norrell$
	($|a_2|/|a_0| = 2$), but fails for $\Strange$, whose two nonzero
	amplitudes share the modulus $1/\sqrt 2$.
	\begin{prop}\label{prop:moulton-qutrit}
		Let $\ket\psi = a_0 \ket 0 + a_1 \ket 1 + a_2 \ket 2 \in \complex^3$
		be a single-qutrit state for which there exist indices
		$i, j \in \{0, 1, 2\}$ with $a_i \ne 0$, $a_j \ne 0$, and
		$|a_i| / |a_j| \ge 2$. Then
		\begin{equation}\label{eq:moulton-qutrit}
			\chiR\!\left(\ket\psi^{\otimes m}\right)
			\ge \frac{m+1}{3 \log_2 (m + 1)}
			\in \Omega\!\left(\frac{m}{\log m}\right).
		\end{equation}
		In particular,
		$\chiR(\Hth^{\otimes m}), \chiR(\Norrell^{\otimes m})
		\in \Omega(m / \log m)$.
	\end{prop}
	
	\begin{proof}
		A qutrit stabilizer state $\ket\sigma$ on $m$ qutrits has the
		canonical form
		\begin{equation}
			\ket\sigma = 3^{-k/2} \sum_{y \in \F{3}^k}
			\omega^{Q(y)} \ket{x_0 + W y},
		\end{equation}
		for some $k \in \{0, \ldots, m\}$, affine offset $x_0$, isometric
		embedding $W \in \F{3}^{m \times k}$, and quadratic form
		$Q : \F{3}^k \to \F{3}$, with $\omega = e^{2 \pi i / 3}$. Each
		coordinate of $\ket\sigma$ therefore lies in
		$\{0, 3^{-k/2} \omega^j : j \in \{0, 1, 2\}\}$. Given a
		decomposition
		$\ket\psi^{\otimes m} = \sum_{i=1}^r c_i \ket{\sigma_i}$ of length
		$r = \chiR(\ket\psi^{\otimes m})$, write $\tilde c_i = 3^{-k_i/2} c_i$;
		then for every coordinate $x \in \F{3}^m$,
		\begin{equation}\label{eq:coord-subset-sum}
			\braket{x}{\psi^{\otimes m}}
			= \sum_{i = 1}^{r} \varepsilon_{i, x} \tilde c_i.
		\end{equation}
		The coefficients $\varepsilon_{i, x}$ take values in
		$\{0, 1, \omega, \omega^2\}$, so each coordinate of
		$\ket\psi^{\otimes m}$ is a subset-sum of the $3r$-tuple
		$(\tilde c_i, \omega \tilde c_i, \omega^2 \tilde c_i)_{i = 1}^{r}$,
		yielding a subset-sum representation of length $3r$ of the
		coordinate vector. The factor of $3$ counts the three cube-roots of
		unity $\{1, \omega, \omega^2\}$ appearing as $\varepsilon_{i, x}$
		values; the qubit version of this argument lifts each
		$\tilde c_i$ to four entries
		$\{\pm \tilde c_i, \pm i \tilde c_i\}$ and yields a $4r$-tuple
		instead. By the subset-sum lower bound
		of~\cite[Theorem~3.1]{lovitz2022techniques}, any subset-sum
		representation of a $q$-tuple containing an exponentially
		increasing subsequence of length $\ell$ (i.e., $\ell$ entries whose
		absolute values each at least double the previous) has length at
		least $\ell / \log_2 \ell$. Hence $3 r \ge \ell / \log_2 \ell$,
		giving $\chiR(\ket\psi^{\otimes m}) \ge \ell / (3 \log_2 \ell)$.
		
		It remains to produce an exponentially increasing subsequence of
		length $\ell = m + 1$ in the coordinates of $\ket\psi^{\otimes m}$.
		By hypothesis, there exist indices $i_a, i_b \in \{0, 1, 2\}$ with
		$a = |a_{i_a}|$, $b = |a_{i_b}|$, $a \ge 2 b > 0$. For each
		$k \in \{0, \ldots, m\}$, set
		\begin{equation}\label{eq:xk-coord}
			x^{(k)} = (\underbrace{i_a, \ldots, i_a}_{k},			\underbrace{i_b, \ldots, i_b}_{m - k}) \in \F{3}^m,
			\qquad\text{so that}\qquad
			\bigl| \braket{x^{(k)}}{\psi^{\otimes m}} \bigr| = a^k b^{m - k}.
		\end{equation}
		The $m + 1$ values $\{ a^k b^{m - k} \}_{k = 0}^{m}$ are distinct
		and form an exponentially increasing sequence of consecutive ratio
		$a / b \ge 2$, so we may take $\ell = m + 1$
		in~\eqref{eq:moulton-qutrit}.
	\end{proof}
	
	The proof above is a direct adaptation of the qubit-$T$-state
	argument of~\cite[Theorem~3.1]{lovitz2022techniques}: the only
	technical change is the factor of $3$ in place of $4$ in the
	subset-sum coefficient count, accounting for cube-roots of unity in
	the qutrit canonical form rather than fourth-roots. What is new
	here is the state-specific verification: which qutrit orbits
	satisfy the modulus-ratio hypothesis ($\Hth$ and $\Norrell$
	do, as recorded above), and which provably do not. The
	$\Strange$-orbit exclusion is the latter; we record it as a
	separate rigidity proposition below.
	
	\begin{prop}\label{prop:strange-rigidity}
		Let $\mathcal{O}_S \subset \complex^3$ denote the projective
		Clifford orbit of $\Strange$ under $\Cl(1, 3)$. Then:
		(i)~every $\ket\phi \in \mathcal{O}_S$ has support of cardinality
		$2$ with the two nonzero amplitudes of equal modulus; and
		(ii)~$|\mathcal{O}_S| = 9$, with the orbit consisting of the three
		nontrivial $2$-element supports of $\F{3}$ and three discrete phase
		choices per support.
		Consequently, the modulus-ratio hypothesis of
		Proposition~\ref{prop:moulton-qutrit} fails at every Clifford
		representative of $\Strange$, and
		Proposition~\ref{prop:moulton-qutrit} yields no asymptotic lower
		bound on $\chiR(\Strange^{\otimes m})$.
	\end{prop}

	The full proof is in Appendix~\ref{app:strange-rigidity}.
	The same property that prevents the subset-sum argument for
	$\Strange$ is what makes $\Strange$ admit the smallest known
	exact-rank upper-bound exponent across the four orbits
	($\gamma_{\mathbb{S}} \le \log_3(2)/2$, by Theorem~\ref{thm:strange-m2});
	closing the gap for $\Strange$ requires a different lower-bound
	technique. The asymptotic growth
	$\Omega(m / \log m)$ established for $\Hth$ and $\Norrell$ above is
	weaker than the $\Omega(m)$ lower bound
	of~\cite{labib2022stabilizer} for a specific qudit-$T$-gate magic
	state, but applies to representatives of orbits that prior work
	does not cover. The algebraic-symmetry arguments underlying
	Proposition~\ref{prop:strange-rigidity} are reminiscent of the
	stabilizer-polytope symmetry analysis
	of~\cite{heinrich2019robustness}, which exploits the same
	Clifford-orbit structure to compute robustness-of-magic on small
	qudit systems.
	
	\section{Magic-state injection}\label{sec:injection}
	
	The stabilizer-rank exponents of Section~\ref{sec:bounds} bound the
	state-vector cost of the deterministic stabilizer-rank simulator
	of~\cite{bravyi2016trading} applied directly to $\ket M^{\otimes m}$.
	Converting them to a circuit-runtime advantage on a concrete
	Clifford-plus-$\ket M$ circuit additionally requires a model for
	consuming the magic-state ancilla. The standard choice is a
	\emph{deterministic injection gadget}: a fixed entangling two-qutrit
	Clifford that, on a chosen ancilla measurement outcome, applies a
	non-Clifford gate to the data and corrects the remaining outcomes
	via Clifford byproducts. For $T_3$ the diagonal qutrit $T$-rotation
	$\mathrm{diag}(1, \omega_9, \omega_9^2)$ supplies one, so
	$\gamma_{T_3} \le 1/2$ is a circuit-runtime bound.
	
	For the three non-$T_3$ orbits, standard single-shot injection
	via a diagonal gate requires the magic state to be a
	\emph{phase state} (i.e.\ have uniform amplitudes up to
	phases)~\cite{anwar2012qutrit, howard2012qudit,
	cui2017diagonal}. Since $\Strange$, $\Hth$, and $\Norrell$ all
	have non-uniform amplitudes, this route is
	closed~\cite{prakash2020magic, prakash2020contextual}.

	\subsection{Two-copy probabilistic conversion}\label{ssec:two-copy}

	A milder
	construction~\cite{bravyi2005universal, anwar2012qutrit},
	introduced for the qubit phase state $(\ket 0 + e^{i\pi/4} \ket 1)/\sqrt 2$
	(the standard ``T-state'' of the qubit distillation literature, in
	our convention an $H$-type representative) and adapted to the
	qutrit Hadamard-eigenstate, instead feeds two copies of $\ket{M}$
	into a two-qutrit Clifford, measures the second qutrit, and
	recovers a Clifford-equivalent phase state on the first qutrit on
	a designated branch. The recovered phase state is
	then injected via the standard single-shot gadget. The resulting
	pipeline consumes $2/q$ copies of $\ket{M}$ in expectation per
	non-Clifford gate, where $q$ is the branch success probability:
	an $O(1)$ ancilla overhead provided (a) such a Clifford exists
	and (b) the recovered phase state implements a non-Clifford gate
	under standard injection. The phase state that emerges in
	Theorem~\ref{thm:two-copy} below implements one specific
	non-Clifford diagonal gate per orbit (the $\Hth$ and $\Norrell$
	gates listed in~\eqref{eq:H3-output} and~\eqref{eq:N-output}); reaching other
	non-Clifford diagonal targets requires an additional
	gate-synthesis step on top of the two-copy primitive.

	A complete enumeration over $\Sp(4, \F{3}) \times \F{3}$ (the
	choice of Clifford and choice of measurement outcome, modulo the
	Heisenberg--Weyl prefactor as in
	Appendix~\ref{app:gadget-reduction}) settles the existence question
	for every non-$T_3$ orbit:

	\begin{theorem}\label{thm:two-copy}
		For $\ket{M} \in \{\Hth, \Norrell\}$ there exist two-qutrit
		Cliffords $C$ and measurement branches $k$ such that
		$\bra{k}_{\mathrm{anc}} C (\ket{M} \otimes \ket{M})$ is proportional
		to a phase state whose standard single-shot injection yields a
		non-Clifford diagonal gate. For $\ket{M} = \Strange$ no such
		protocol exists: every $C$ that maps $\Strange \otimes \Strange$
		to a phase state injects a Clifford gate.
	\end{theorem}

	\begin{proof}
		By exhaustive search over the symplectic-quotient representation (Lemma~\ref{lem:gadget-reduction}), one identifies valid protocols for $\Hth$ and $\Norrell$. The raw search outputs the following two-qutrit Cliffords:
		\begin{equation}\label{eq:two-copy-protocols}
			\begin{aligned}
				C_{\Hth}    &= H_1 \mathrm{SUM} H_1 \mathrm{SUM} H_1 S_2 \mathrm{SUM} H_1 H_2^2, \\
				C_{\Norrell} &= H_1 H_2 \mathrm{SUM} H_1^\dagger H_2 \mathrm{SUM} H_1.
			\end{aligned}
		\end{equation}

		\paragraph{Hadamard-eigenstate protocol.}
		Let $N = \sqrt{3 - \sqrt{3}}$ and $c = (\sqrt{3}-1)/2$ (the same constants used throughout Appendix~\ref{app:h3}). The two-copy input state is $\Hth \otimes \Hth$, where the single-copy Hadamard-eigenstate is $\Hth = \frac{1}{N} \big(\ket{0} + c\ket{1} + c\ket{2}\big)$. Applying the two-qutrit Clifford $C_{\Hth}$ to this input state and projecting the ancilla (leg 2) onto the outcome $k = 1$ yields the data-qutrit output vector
		\begin{equation}\label{eq:H3-output}
			\ket{\psi_{k=1}} = \bra{1}_{\mathrm{anc}} C_{\Hth} (\Hth \otimes \Hth) = \frac{\lambda}{\sqrt{3}} \big(\ket{0} + e^{i\pi/2}\ket{1} + e^{i\pi/3}\ket{2}\big),
		\end{equation}
		where $\lambda \in \complex$ is a normalization scalar satisfying $|\lambda|^2 = 3/8$. Since all three amplitudes of the post-selected state have equal modulus, it is a phase state. The relative phases are $\arg(\psi_1 / \psi_0) = \pi/2$ and $\arg(\psi_2 / \psi_0) = \pi/3$, corresponding to the single-qutrit diagonal gate $\mathrm{diag}(1, i, e^{i\pi/3})$. Since a single-qutrit diagonal Clifford gate must have phases that are multiples of $2\pi/3$, the phase $\pi/2$ certifies that the injected gate is non-Clifford.

		\paragraph{Norrell protocol.}
		The two-copy input state is $\Norrell \otimes \Norrell$, where $\Norrell = \frac{1}{\sqrt{6}}(\ket{0} + \ket{1} - 2\ket{2})$ in the computational basis. First, the gate $H_1$ acts on the data qutrit as the qutrit Fourier transform, leaving the ancilla unchanged. On the amplitudes $v_{a,b} = c_N(a) c_N(b)/6$ (where $c_N = (1, 1, -2)^T$), the data Fourier transform produces
		\begin{equation}
			\ket{\Psi_1} = (H_1 \otimes \I) (\Norrell \otimes \Norrell) = \frac{1}{6\sqrt{3}} \sum_{j, b} \hat{c}(j) c_N(b) \ket{j, b},
		\end{equation}
		where the inner sum evaluates using $1 + \omega + \omega^2 = 0$ as
		\begin{equation}
			\hat{c}(j) = \sum_{a=0}^2 \omega^{aj} c_N(a) = 1 + \omega^j - 2\omega^{2j}.
		\end{equation}
		This evaluates to $\hat{c}(0) = 0$, $\hat{c}(1) = -3\omega^2$, and $\hat{c}(2) = -3\omega$. Thus, after $H_1$, only the data components with $j \in \{1, 2\}$ survive.

		Applying the remaining gates of $C_{\Norrell}$ and projecting onto the ancilla outcome $k = 0$ yields
		\begin{equation}\label{eq:N-output}
			\ket{\psi_{k=0}} = \bra{0}_{\mathrm{anc}} C_{\Norrell} (\Norrell \otimes \Norrell) = \frac{\mu}{\sqrt{3}} \big(\ket{0} - \ket{1} - \ket{2}\big),
		\end{equation}
		with $|\mu|^2 = 1/4$, confirming a success probability of $1/4$ to obtain a phase state. The relative phases are both $\pi$, defining the injected diagonal gate $\mathrm{diag}(1, -1, -1)$. Because $\pi$ is not a multiple of $2\pi/3$, this gate is non-Clifford.

		Finally, for $\ket{M} = \Strange$, an exhaustive search confirms that no two-copy conversion protocol exists: every $(C, k)$ pair that maps $\Strange \otimes \Strange$ to a phase state yields relative phases restricted to the Clifford-trivial ones.
	\end{proof}

	Unlike the qubit case~\cite[Sec.~III]{bravyi2005universal}, where the Hadamard gate naturally pair-conjugates real amplitudes to equalize their moduli, the qutrit Fourier transform has no such conjugate-pairing property. Modulus equalization for $\Hth$ and $\Norrell$ therefore requires the longer alternating sequences of Hadamards, SUMs, and phase gates in~\eqref{eq:two-copy-protocols}.

	Operationally, Theorem~\ref{thm:two-copy} provides a constant-overhead route from copies of $\ket M \in \{\Hth, \Norrell\}$ to one specific non-Clifford diagonal gate per orbit ($\mathrm{diag}(1, i, e^{i\pi/3})$ and $\mathrm{diag}(1, -1, -1)$ respectively). Arbitrary diagonal targets outside this image require additional gate-synthesis overhead. 
	
	For $\Strange$, the best upper-bound exponent ($\gamma_{\mathbb{S}} \le \log_3(2)/2$) remains operationally inaccessible under deterministic single-shot injection, two-copy probabilistic conversion and even three-copy protocols: an exhaustive search over the 3-qutrit symplectic group using coset pruning identifies no protocols to obtain a non-Clifford diagonal gate.
	
	\section{The \texttt{stabrank} library}\label{sec:stabrank}

	The numerical and verification work in this paper is carried out
	in \texttt{stabrank}, an open-source Python package with a C++
	core released with the manuscript. The library has a heuristic
	upper-bound stage (simulated-annealing decomposition search,
	hereafter SA) and two exhaustive stages (the $k$-tuple enumeration
	that produces rank lower-bound certificates and the
	symplectic-quotient sweep behind the gadget non-existence result of
	Section~\ref{sec:injection}). The SA search is treated as untrusted
	scaffolding: its output is a candidate canonical-form decomposition
	that is then replayed against an independent verifier; the search
	itself never enters the verification path.
	Table~\ref{tab:claims} maps each paper result to its producing
	artifact and audit channels, and Figure~\ref{fig:workflow}
	sketches the upper-bound pipeline.
	
	\begin{table}[!htbp]
		\centering
		\small
		\begin{tabular}{l l l}
			\toprule
			Result & Producing artifact & Auditing channel(s) \\
			\midrule
			Thm.~\ref{thm:strange-m2}        & SA decomposition          & numerical; Lean~4 \\
			Prop.~\ref{prop:moulton-qutrit}  & analytic proof            & --- \\
			Thm.~\ref{thm:m3-bound} ($\le 4$) & SA decomposition         & numerical; Lean~4 \\
			Thm.~\ref{thm:m3-bound} ($\ge 4$) & exhaustive triple search & enumeration contract; residual gap \\
			Gadget non-existence & $\Sp(4, \F{3})$ sweep & enumeration contract; branch-test replay \\
			Thm.~\ref{thm:two-copy} & two-copy $\Sp(4, \F{3})$ sweep & enumeration contract; gate-sequence replay \\
			Prop.~\ref{prop:strange-rigidity} ($\Strange$ rigidity) & analytic proof & --- \\
			$\Norrell^{\otimes 4} \le 7$ (App.~\ref{app:norrell-m4}) & SA decomposition & numerical; SymPy~\cite{meurer2017sympy}; Lean~4 \\
			\midrule
			$\chiR(\ket{H}^{\otimes 4}) \ge 4$~\cite{bravyi2016trading}  & exhaustive triple search (verification only) & enumeration contract; residual gap \\
			$\chiR(\ket{T}^{\otimes 2}) \le 2$  & SA decomposition & numerical \\
			$\chiR(\ket{T}^{\otimes 3}) \le 3$  & SA decomposition & numerical \\
			$\chiR(\ket{T}^{\otimes 4}) = 3$  & algebraic identity (App.~\ref{app:qubit-t}) & numerical \\
			\bottomrule
		\end{tabular}
		\caption{Result-to-artifact-to-verifier map for the paper's key claims, indicating the primary search method and validation channels (numerical, exact-rational, and Lean~4 machine-checking).}
		\label{tab:claims}
	\end{table}
	
	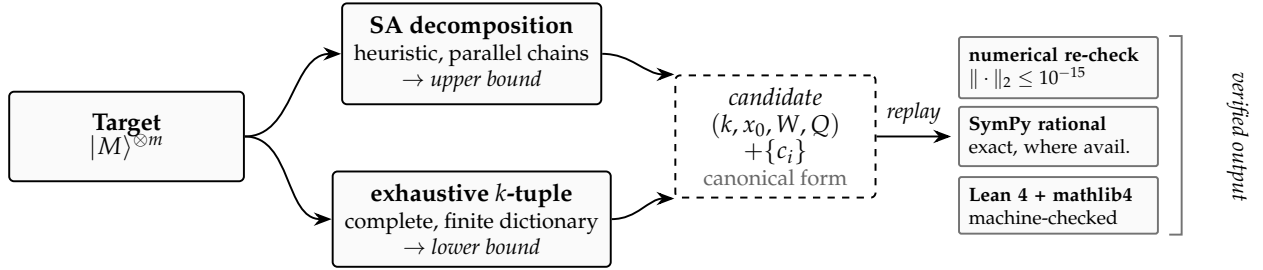
\begin{figure}[!htbp]
		\centering
		\resizebox{\textwidth}{!}{%
			\begin{tikzpicture}[
				font=\small,
				>={Stealth[length=2.6mm]},
				box/.style={draw, thick, rounded corners=2pt,
					fill=black!2, align=center, inner sep=5pt,
					minimum width=3.4cm, minimum height=1.3cm},
				cand/.style={draw, thick, rounded corners=2pt, dashed,
					fill=white, align=center, inner sep=5pt,
					minimum width=2.9cm, minimum height=1.6cm},
				vcell/.style={draw=black!55, thick, rounded corners=1pt,
					fill=black!2, align=left, inner sep=4pt,
					text width=2.6cm, minimum height=0.7cm,
					font=\scriptsize},
				arrlabel/.style={font=\footnotesize\itshape, midway, above, yshift=1pt}
				]
				\node[box] (target) at (0,0)
				{\textbf{Target}\\[-2pt]
					$\ket{M}^{\otimes m}$};
				
				\node[box] (s1) at (5, 1.2)
				{\textbf{SA decomposition}\\[-1pt]
					\footnotesize heuristic, parallel chains\\[-1pt]
					\footnotesize\itshape $\to$ upper bound};
				\node[box] (s2) at (5,-1.2)
				{\textbf{exhaustive $k$-tuple}\\[-1pt]
					\footnotesize complete, finite dictionary\\[-1pt]
					\footnotesize\itshape $\to$ lower bound};
				
				\node[cand] (cand) at (9.4, 0)
				{\textit{candidate}\\[-1pt]
					$(k, x_0, W, Q)$\\[-1pt]
					$+\{c_i\}$\\[-1pt]
					\footnotesize\color{black!60}canonical form};
				
				\node[vcell] (v1) at (13.5, 1)
				{\textbf{numerical re-check}\\$\|\cdot\|_2 \le 10^{-15}$};
				\node[vcell] (v2) at (13.5, 0)
				{\textbf{SymPy rational}\\exact, where avail.};
				\node[vcell] (v3) at (13.5,-1)
				{\textbf{Lean 4 + mathlib4}\\machine-checked};
				
				\draw[->,thick] (target.east) to[out=0,in=180] (s1.west);
				\draw[->,thick] (target.east) to[out=0,in=180] (s2.west);
				
				\draw[->,thick] (s1.east) to[out=0,in=180] (cand.north west);
				\draw[->,thick] (s2.east) to[out=0,in=180] (cand.south west);
				
				\draw[->,thick] (cand.east) -- node[arrlabel] {replay} ($(v2.west)+(-0.1,0)$);
				
				\draw[black!55, thick] ($(v1.north east)+(0.15,0)$)
				-- ++(0.18,0)
				-- ($(v3.south east)+(0.33,0)$)
				-- ++(-0.18,0);
				\node[font=\footnotesize\itshape, rotate=-90, anchor=south]
				at ($(v2.east)+(0.85,0)$) {verified output};
			\end{tikzpicture}%
		}
		\caption{The \texttt{stabrank} decomposition pipeline
			(upper-bound side). Simulated annealing returns a candidate
			canonical-form decomposition $(k, x_0, W, Q)$ with coefficients
			$\{c_i\}$, replayed through three independent verifiers:
			numerical at machine precision (all seven qutrit identities),
			exact-rational in SymPy~\cite{meurer2017sympy} (the heaviest
			identity, $\Norrell^{\otimes 4}$), and Lean~4 + mathlib4 (all
			seven qutrit identities). The exhaustive $k$-tuple enumeration
			and the symplectic-quotient sweep
			(Section~\ref{sec:injection}) instead emit JSON certificates
			audited via the path in Section~\ref{sec:stabrank}.}
		\label{fig:workflow}
	\end{figure}
	
	The Lean~4 formalizations re-derive the canonical-form amplitude
	identities from the parameters of Appendix~\ref{app:decomps}
	alone, with no dependency on the search or its floating-point
	residual. Soundness of the exhaustive certificates rests on a
	large residual gap: across every published certificate the
	smallest non-witness least-squares residual is $\approx 0.14$,
	nine orders of magnitude above the $10^{-10}$ tolerance used to
	flag a true witness. The $\Sp(4, \F{3})$ gadget sweep uses
	$\mathtt{atol} = 10^{-8}$ for proportionality-to-unitary and
	$\mathtt{atol} = 10^{-5}$ for Clifford-equivalence against the
	$216$-element single-qutrit Clifford group; recovering the $T_3$
	injection gadget as a positive control confirms calibration.
	
	\section{Conclusions}\label{sec:conclusions}

	We give the first nontrivial exact-rank upper-bound exponents for
	the three non-$T_3$ qutrit orbits ($\gamma_{\mathbb{S}} \le \log_3(2)/2$,
	$\gamma_{H_3}, \gamma_{\mathbb{N}} \le \log_3(4)/3$, all strictly
	below the $\gamma_{T_3} \le 1/2$ baseline
	of~\cite{kocia2021improved}), exhaustive small-$m$ tight values
	$\chiR(\ket M^{\otimes 3}) = 4$ for $\ket M \in \{\Strange, \Hth,
	\Norrell\}$, and the first $\Omega(m/\log m)$ asymptotic lower
	bounds for $\Hth$ and $\Norrell$ via a qutrit adaptation
	of~\cite{lovitz2022techniques}. Operationally, $\Hth$ and
	$\Norrell$ admit explicit two-copy probabilistic conversion to
	injectable phase states. In contrast, $\Strange$ is rigid:
	exhaustive searches over the 2-qutrit and 3-qutrit symplectic groups
	(using pruning) confirm that no conversion protocol exists.
	Thus, the smaller upper-bound exponent of $\Strange$ remains
	operationally inaccessible under low-copy consumption.

	An interesting open direction is to establish asymptotic lower bounds
	for $\Strange$, where the subset-sum obstruction technique
	used for the other qutrit orbits does not directly apply.
	Standard stabilizer-nullity or dyadic-monotone techniques
	(e.g., following~\cite{beverland2020lower}) might offer a starting
	point, though their applicability to $\Strange$ remains to be
	fully explored. Similarly, while approximate stabilizer rank lower
	bounds have been studied in the qubit setting (such as the
	probabilistic bounds in~\cite{mehraban2023quadratic}), extending
	these results to the qutrit setting is another potential avenue of
	inquiry. Finally, it remains an open question whether a physical
	Clifford conversion protocol for $\Strange$ is possible for any
	number of copies, or if the stabilizer rank exponent can be made
	operationally useful via alternative consumption models~\cite{raussendorf2001one,
	hall2007cluster, booth2021outcome, mackeprang2022power}.
	
	\section*{Acknowledgements}

	We thank David Gosset for an early look at the manuscript. We used
	the generative AI tool Claude Opus 4.7 during code development for
	the \texttt{stabrank} library, including the Lean~4 + mathlib4
	formalizations of the qutrit decomposition identities of
	Appendix~\ref{app:decomps}. This work was supported by the U.S.
	Department of Energy, Office of Science, Office of Advanced
	Scientific Computing Research, Accelerated Research in Quantum
	Computing under Award Number DE-SC0025336. This material is also
	based upon work supported by the U.S. Department of Energy, Office
	of Science, National Quantum Information Science Research Centers,
	Quantum Science Center. FL acknowledges support from the European
	Union through the QLASS project (EU Horizon Europe grant agreement
	101135876). Views and opinions expressed are however those of the
	authors only and do not necessarily reflect those of the European
	Union. Neither the European Union nor the granting authority can be
	held responsible for them.
	
	\bibliographystyle{alpha}
	\bibliography{refs}

@article{anwar2012qutrit,
  author       = {Anwar, Hussain and Campbell, Earl T. and Browne, Dan E.},
  title        = {Qutrit magic state distillation},
  journal      = {New J. Phys.},
  volume       = {14},
  pages        = {063006},
  year         = {2012},
  doi          = {10.1088/1367-2630/14/6/063006},
  eprint       = {1202.2326},
  archivePrefix= {arXiv},
  primaryClass = {quant-ph},
}

@article{appleby2005symmetric,
  author       = {Appleby, D. M.},
  title        = {Symmetric informationally complete-positive operator valued measures and the extended {C}lifford group},
  journal      = {J. Math. Phys.},
  volume       = {46},
  pages        = {052107},
  year         = {2005},
  doi          = {10.1063/1.1896384},
  eprint       = {quant-ph/0412001},
  archivePrefix= {arXiv},
  primaryClass = {quant-ph},
}

@article{hostens2005stabilizer,
  author       = {Hostens, Erik and Dehaene, Jeroen and De Moor, Bart},
  title        = {Stabilizer states and {C}lifford operations for systems of arbitrary dimensions and modular arithmetic},
  journal      = {Phys. Rev. A},
  volume       = {71},
  pages        = {042315},
  year         = {2005},
  doi          = {10.1103/PhysRevA.71.042315},
  eprint       = {quant-ph/0408190},
  archivePrefix= {arXiv},
  primaryClass = {quant-ph},
}

@misc{booth2021outcome,
  author       = {Booth, Robert I. and Kissinger, Aleks and Markham, Damian},
  title        = {Outcome determinism in measurement-based quantum computation with qudits},
  year         = {2021},
  eprint       = {2109.13810},
  archivePrefix= {arXiv},
  primaryClass = {quant-ph},
}

@article{bravyi2005universal,
  author       = {Bravyi, Sergey and Kitaev, Alexei},
  title        = {Universal quantum computation with ideal {C}lifford gates and noisy ancillas},
  journal      = {Phys. Rev. A},
  volume       = {71},
  pages        = {022316},
  year         = {2005},
  doi          = {10.1103/PhysRevA.71.022316},
  eprint       = {quant-ph/0403025},
  archivePrefix= {arXiv},
  primaryClass = {quant-ph},
}

@article{bravyi2016improved,
  author       = {Bravyi, Sergey and Gosset, David},
  title        = {Improved classical simulation of quantum circuits dominated by {C}lifford gates},
  journal      = {Phys. Rev. Lett.},
  volume       = {116},
  pages        = {250501},
  year         = {2016},
  doi          = {10.1103/PhysRevLett.116.250501},
  eprint       = {1601.07601},
  archivePrefix= {arXiv},
  primaryClass = {quant-ph},
}

@article{bravyi2016trading,
  author       = {Bravyi, Sergey and Smith, Graeme and Smolin, John A.},
  title        = {Trading classical and quantum computational resources},
  journal      = {Phys. Rev. X},
  volume       = {6},
  pages        = {021043},
  year         = {2016},
  doi          = {10.1103/PhysRevX.6.021043},
  eprint       = {1506.01396},
  archivePrefix= {arXiv},
  primaryClass = {quant-ph},
}

@article{bravyi2019simulation,
  author       = {Bravyi, Sergey and Browne, Dan and Calpin, Padraic and Campbell, Earl and Gosset, David and Howard, Mark},
  title        = {Simulation of quantum circuits by low-rank stabilizer decompositions},
  journal      = {Quantum},
  volume       = {3},
  pages        = {181},
  year         = {2019},
  doi          = {10.22331/q-2019-09-02-181},
  eprint       = {1808.00128},
  archivePrefix= {arXiv},
  primaryClass = {quant-ph},
}

@article{campbell2012magic,
  author       = {Campbell, Earl T. and Anwar, Hussain and Browne, Dan E.},
  title        = {Magic-state distillation in all prime dimensions using quantum {R}eed--{M}uller codes},
  journal      = {Phys. Rev. X},
  volume       = {2},
  pages        = {041021},
  year         = {2012},
  doi          = {10.1103/PhysRevX.2.041021},
  eprint       = {1205.3104},
  archivePrefix= {arXiv},
  primaryClass = {quant-ph},
}

@article{hall2007cluster,
  author       = {Hall, William},
  title        = {Cluster state quantum computation for many-level systems},
  journal      = {Phys. Rev. A},
  volume       = {75},
  pages        = {062321},
  year         = {2007},
  doi          = {10.1103/PhysRevA.75.062321},
  eprint       = {quant-ph/0512130},
  archivePrefix= {arXiv},
  primaryClass = {quant-ph},
}

@article{huang2019approximate,
  author       = {Huang, Yifei and Love, Peter},
  title        = {Approximate stabilizer rank and improved weak simulation of {C}lifford-dominated circuits for qudits},
  journal      = {Phys. Rev. A},
  volume       = {99},
  pages        = {052307},
  year         = {2019},
  doi          = {10.1103/PhysRevA.99.052307},
  eprint       = {1808.02406},
  archivePrefix= {arXiv},
  primaryClass = {quant-ph},
}

@article{veitch2014resource,
  author       = {Veitch, Victor and Mousavian, S. A. Hamed and Gottesman, Daniel and Emerson, Joseph},
  title        = {The resource theory of stabilizer quantum computation},
  journal      = {New J. Phys.},
  volume       = {16},
  pages        = {013009},
  year         = {2014},
  doi          = {10.1088/1367-2630/16/1/013009},
  eprint       = {1307.7171},
  archivePrefix= {arXiv},
  primaryClass = {quant-ph},
}

@article{jain2020qutrit,
  author       = {Jain, Akalank and Prakash, Shiroman},
  title        = {Qutrit and ququint magic states},
  journal      = {Phys. Rev. A},
  volume       = {102},
  pages        = {042409},
  year         = {2020},
  doi          = {10.1103/PhysRevA.102.042409},
  eprint       = {2003.07164},
  archivePrefix= {arXiv},
  primaryClass = {quant-ph},
}

@misc{kalra2025stabilizer,
  author       = {Kalra, Amolak Ratan and Sinha, Pulkit},
  title        = {Stabilizer Ranks, {B}arnes-{W}all Lattices and Magic Monotones},
  year         = {2025},
  eprint       = {2503.04101},
  archivePrefix= {arXiv},
  primaryClass = {quant-ph},
}

@article{kocia2021improved,
  author       = {Kocia, Lucas and Sarovar, Mohan},
  title        = {Improved simulation of quantum circuits by fewer {G}aussian eliminations},
  journal      = {Phys. Rev. A},
  volume       = {103},
  pages        = {022603},
  year         = {2021},
  doi          = {10.1103/PhysRevA.103.022603},
  eprint       = {2003.01130},
  archivePrefix= {arXiv},
  primaryClass = {quant-ph},
}

@article{labib2022stabilizer,
  author       = {Labib, Farrokh},
  title        = {Stabilizer rank and higher-order {F}ourier analysis},
  journal      = {Quantum},
  volume       = {6},
  pages        = {645},
  year         = {2022},
  doi          = {10.22331/q-2022-02-09-645},
  eprint       = {2107.10551},
  archivePrefix= {arXiv},
  primaryClass = {quant-ph},
}

@article{lovitz2022techniques,
  author       = {Lovitz, Benjamin and Steffan, Vincent},
  title        = {New techniques for bounding stabilizer rank},
  journal      = {Quantum},
  volume       = {6},
  pages        = {692},
  year         = {2022},
  doi          = {10.22331/q-2022-04-20-692},
  eprint       = {2110.07781},
  archivePrefix= {arXiv},
  primaryClass = {quant-ph},
}

@article{peleg2022lower,
  author       = {Peleg, Shir and Shpilka, Amir and Volk, Ben Lee},
  title        = {Lower Bounds on Stabilizer Rank},
  journal      = {Quantum},
  volume       = {6},
  pages        = {652},
  year         = {2022},
  doi          = {10.22331/q-2022-02-15-652},
  eprint       = {2106.03214},
  archivePrefix= {arXiv},
  primaryClass = {quant-ph},
}

@misc{mackeprang2022power,
  author       = {Mackeprang, Jelena and Bhatti, Daniel and Hoban, Matty J.},
  title        = {The power of qutrits for non-adaptive measurement-based quantum computing},
  year         = {2022},
  eprint       = {2203.12411},
  archivePrefix= {arXiv},
  primaryClass = {quant-ph},
}

@misc{mehraban2023quadratic,
  author       = {Mehraban, Saeed and Tahmasbi, Mehrdad},
  title        = {Quadratic Lower bounds on the Approximate Stabilizer Rank: A Probabilistic Approach},
  year         = {2023},
  eprint       = {2305.10277},
  archivePrefix= {arXiv},
  primaryClass = {quant-ph},
}

@inproceedings{mathlib2020,
  author       = {{The mathlib Community}},
  title        = {The {L}ean mathematical library},
  booktitle    = {Proceedings of the 9th ACM SIGPLAN International Conference on Certified Programs and Proofs},
  series       = {CPP 2020},
  pages        = {367--381},
  year         = {2020},
  doi          = {10.1145/3372885.3373824},
}

@inproceedings{moura2021lean4,
  author       = {de Moura, Leonardo and Ullrich, Sebastian},
  title        = {The {L}ean 4 Theorem Prover and Programming Language},
  booktitle    = {Automated Deduction -- CADE 28},
  pages        = {625--635},
  year         = {2021},
  doi          = {10.1007/978-3-030-79876-5_37},
}

@article{prakash2020magic,
  author       = {Prakash, Shiroman},
  title        = {Magic state distillation with the ternary {G}olay code},
  journal      = {Proc. R. Soc. A},
  volume       = {476},
  pages        = {20200187},
  year         = {2020},
  doi          = {10.1098/rspa.2020.0187},
  eprint       = {2003.02717},
  archivePrefix= {arXiv},
  primaryClass = {quant-ph},
}

@phdthesis{qassim2020classical,
  author       = {Qassim, Hammam},
  title        = {Classical Simulations of Quantum Systems Using Stabilizer Decompositions},
  school       = {University of Waterloo},
  year         = {2020},
}

@article{qassim2021improved,
  author       = {Qassim, Hammam and Pashayan, Hakop and Gosset, David},
  title        = {Improved upper bounds on the stabilizer rank of magic states},
  journal      = {Quantum},
  volume       = {5},
  pages        = {606},
  year         = {2021},
  doi          = {10.22331/q-2021-12-20-606},
  eprint       = {2106.07740},
  archivePrefix= {arXiv},
  primaryClass = {quant-ph},
}

@article{raussendorf2001one,
  author       = {Raussendorf, Robert and Briegel, Hans J.},
  title        = {A one-way quantum computer},
  journal      = {Phys. Rev. Lett.},
  volume       = {86},
  pages        = {5188--5191},
  year         = {2001},
  doi          = {10.1103/PhysRevLett.86.5188},
}

@article{meurer2017sympy,
  author       = {Meurer, Aaron and Smith, Christopher P. and Paprocki, Mateusz
                  and {\v{C}}ert{\'i}k, Ond{\v{r}}ej and Kirpichev, Sergey B.
                  and Rocklin, Matthew and Kumar, Amit and Ivanov, Sergiu
                  and Moore, Jason K. and Singh, Sartaj and Rathnayake, Thilina
                  and Vig, Sean and Granger, Brian E. and Muller, Richard P.
                  and Bonazzi, Francesco and Gupta, Harsh and Vats, Shivam
                  and Johansson, Fredrik and Pedregosa, Fabian and Curry, Matthew J.
                  and Terrel, Andy R. and Rou{\v{c}}ka, {\v{S}}t{\v{e}}p{\'a}n
                  and Saboo, Ashutosh and Fernando, Isuru and Kulal, Sumith
                  and Cimrman, Robert and Scopatz, Anthony},
  title        = {{SymPy}: {S}ymbolic computing in {P}ython},
  journal      = {PeerJ Computer Science},
  volume       = {3},
  pages        = {e103},
  year         = {2017},
  doi          = {10.7717/peerj-cs.103},
}

@article{zhu2010sic,
  author       = {Zhu, Huangjun},
  title        = {{SIC} {POVM}s and {C}lifford groups in prime dimensions},
  journal      = {J. Phys. A: Math. Theor.},
  volume       = {43},
  pages        = {305305},
  year         = {2010},
  doi          = {10.1088/1751-8113/43/30/305305},
}

@article{gottesman1999fault,
  author       = {Gottesman, Daniel},
  title        = {Fault-tolerant quantum computation with higher-dimensional systems},
  journal      = {Chaos, Solitons \& Fractals},
  volume       = {10},
  pages        = {1749--1758},
  year         = {1999},
  doi          = {10.1016/S0960-0779(98)00218-5},
  eprint       = {quant-ph/9802007},
  archivePrefix= {arXiv},
  primaryClass = {quant-ph},
}

@misc{gottesman1998heisenberg,
  author       = {Gottesman, Daniel},
  title        = {The {H}eisenberg representation of quantum computers},
  year         = {1998},
  eprint       = {quant-ph/9807006},
  archivePrefix= {arXiv},
  primaryClass = {quant-ph},
}

@article{aaronson2004improved,
  author       = {Aaronson, Scott and Gottesman, Daniel},
  title        = {Improved simulation of stabilizer circuits},
  journal      = {Phys. Rev. A},
  volume       = {70},
  pages        = {052328},
  year         = {2004},
  doi          = {10.1103/PhysRevA.70.052328},
  eprint       = {quant-ph/0406196},
  archivePrefix= {arXiv},
  primaryClass = {quant-ph},
}

@article{veitch2012negative,
  author       = {Veitch, Victor and Ferrie, Christopher and Gross, David and Emerson, Joseph},
  title        = {Negative quasi-probability as a resource for quantum computation},
  journal      = {New J. Phys.},
  volume       = {14},
  pages        = {113011},
  year         = {2012},
  doi          = {10.1088/1367-2630/14/11/113011},
  eprint       = {1201.1256},
  archivePrefix= {arXiv},
  primaryClass = {quant-ph},
}

@article{howard2014contextuality,
  author       = {Howard, Mark and Wallman, Joel and Veitch, Victor and Emerson, Joseph},
  title        = {Contextuality supplies the `magic' for quantum computation},
  journal      = {Nature},
  volume       = {510},
  pages        = {351--355},
  year         = {2014},
  doi          = {10.1038/nature13460},
  eprint       = {1401.4174},
  archivePrefix= {arXiv},
  primaryClass = {quant-ph},
}

@article{howard2012qudit,
  author       = {Howard, Mark and Vala, Jiri},
  title        = {Qudit versions of the qubit $\pi/8$ gate},
  journal      = {Phys. Rev. A},
  volume       = {86},
  pages        = {022316},
  year         = {2012},
  doi          = {10.1103/PhysRevA.86.022316},
  eprint       = {1206.1598},
  archivePrefix= {arXiv},
  primaryClass = {quant-ph},
}

@article{dawkins2015qutrit,
  author       = {Dawkins, Hillary and Howard, Mark},
  title        = {Qutrit magic state distillation tight in some directions},
  journal      = {Phys. Rev. Lett.},
  volume       = {115},
  pages        = {030501},
  year         = {2015},
  doi          = {10.1103/PhysRevLett.115.030501},
  eprint       = {1504.05965},
  archivePrefix= {arXiv},
  primaryClass = {quant-ph},
}

@article{prakash2020contextual,
  author       = {Prakash, Shiroman and Gupta, Aashi},
  title        = {Contextual bound states for qudit magic state distillation},
  journal      = {Phys. Rev. A},
  volume       = {101},
  pages        = {010303},
  year         = {2020},
  doi          = {10.1103/PhysRevA.101.010303},
  eprint       = {1905.00392},
  archivePrefix= {arXiv},
  primaryClass = {quant-ph},
}

@article{beverland2020lower,
  author       = {Beverland, Michael and Campbell, Earl and Howard, Mark and Kliuchnikov, Vadym},
  title        = {Lower bounds on the non-{C}lifford resources for quantum computations},
  journal      = {Quantum Sci. Technol.},
  volume       = {5},
  pages        = {035009},
  year         = {2020},
  doi          = {10.1088/2058-9565/ab8963},
  eprint       = {1904.01124},
  archivePrefix= {arXiv},
  primaryClass = {quant-ph},
}

@article{heinrich2019robustness,
  author       = {Heinrich, Markus and Gross, David},
  title        = {Robustness of magic and symmetries of the stabiliser polytope},
  journal      = {Quantum},
  volume       = {3},
  pages        = {132},
  year         = {2019},
  doi          = {10.22331/q-2019-04-08-132},
  eprint       = {1807.10296},
  archivePrefix= {arXiv},
  primaryClass = {quant-ph},
}

@article{cui2017diagonal,
  author       = {Cui, Shawn X. and Gottesman, Daniel and Krishna, Anirudh},
  title        = {Diagonal gates in the {C}lifford hierarchy},
  journal      = {Phys. Rev. A},
  volume       = {95},
  pages        = {012329},
  year         = {2017},
  doi          = {10.1103/PhysRevA.95.012329},
  eprint       = {1608.06596},
  archivePrefix= {arXiv},
  primaryClass = {quant-ph},
}

@misc{stabrank-software,
  author       = {Labib, Farrokh and Russo, Vincent},
  title        = {\texttt{stabrank}: a library for stabilizer-rank decompositions and certificates},
  year         = {2026},
  howpublished = {\url{https://github.com/unitaryfoundation/stabrank}},
}
	
	\appendix
	\section{Symplectic-quotient reduction}\label{app:gadget-reduction}

	Reduction of the search domain
	$\Cl(2, 3) = H(2, 3) \rtimes \Sp(4, \F{3})$ to its symplectic
	quotient, used in the gadget non-existence result and
	Theorem~\ref{thm:two-copy}. For an entangling two-qutrit Clifford
	$C \in \Cl(2, 3)$ and ancilla outcome $k \in \F{3}$, write
	\begin{equation}\label{eq:Ek-def}
		E_k(C, \ket M) =
		(\I_{\mathrm{data}} \otimes \bra{k}_{\mathrm{anc}}) C		(\I_{\mathrm{data}} \otimes \ket{M})
		\in \L(\complex^3)
	\end{equation}
	for the unnormalized data-register linear operator implemented on
	measurement branch $k$. In the two-copy conversion case of
	Theorem~\ref{thm:two-copy} the data register is also prepared in
	$\ket M$, so the branch-$k$ data-register output state is the
	vector $E_k(C, \ket M) \ket M \in \complex^3$.

	\begin{lemma}\label{lem:gadget-reduction}
		Fix $\ket{M} \in \complex^3$ and write $C_{\mathrm{ent}} =
		D C_{\mathrm{sp}}$ with $D \in H(2, 3)$ and $C_{\mathrm{sp}} \in
		\Sp(4, \F{3})$. Decompose $D$ into its data-leg and ancilla-leg
		factors $D = D_{\mathrm{data}} \otimes D_{\mathrm{anc}}$, with
		$D_{\mathrm{anc}} = X^{a} Z^{b}$ for some $(a, b) \in \F{3}^2$.
		Then for every $k \in \F{3}$,
		\begin{equation}\label{eq:gadget-reduction}
			E_k(C_{\mathrm{ent}}, \ket{M})
			= \omega_3^{b(k - a)} D_{\mathrm{data}} E_{k - a}(C_{\mathrm{sp}}, \ket{M}).
		\end{equation}
		Consequently, the existence of a valid deterministic injection
		gadget for $\ket{M}$ depends only on the $\Sp(4, \F{3})$ coset of
		$C_{\mathrm{ent}}$; the same applies to the existence of a valid
		two-copy conversion protocol in the sense of
		Section~\ref{ssec:two-copy}.
	\end{lemma}

	\begin{proof}
		Using the identity
		$\bra{k} X^{a} Z^{b} = \omega_3^{b(k - a)} \bra{k - a}$,
		\begin{equation}
			\begin{aligned}
				E_k(C_{\mathrm{ent}}, \ket{M})
				&= \bra{k}_{\mathrm{anc}} (D_{\mathrm{data}} \otimes D_{\mathrm{anc}})
				 C_{\mathrm{sp}} (\I \otimes \ket{M}) \\
				&= D_{\mathrm{data}} (\bra{k} D_{\mathrm{anc}})_{\mathrm{anc}}
				 C_{\mathrm{sp}} (\I \otimes \ket{M}) \\
				&= \omega_3^{b(k - a)} D_{\mathrm{data}} E_{k - a}(C_{\mathrm{sp}}, \ket{M}).
			\end{aligned}
		\end{equation}
		The prefactor $D_{\mathrm{data}}$ is a single-qutrit Pauli, and
		$k \mapsto k - a$ permutes $\F{3}$, so the multisets
		$\{E_k(C_{\mathrm{ent}}, M)\}_{k}$ and $\{E_k(C_{\mathrm{sp}},
		M)\}_{k}$ coincide up to a global Pauli prefactor and
		branch-dependent phases. Proportionality to a unitary, the
		Clifford-byproduct condition, and the
		Clifford-equivalent-to-phase-state condition are all invariant
		under both transformations.
	\end{proof}

	The reduction shows that whether a valid gadget exists is a
	property of the $\Sp(4, \F{3})$ coset, not of any particular
	representative. It does not say anything about gate count: the
	circuits in~\eqref{eq:two-copy-protocols} are the symplectic
	representatives returned by the search, and a different coset
	representative might admit a shorter decomposition.

	\section{Proof of Proposition~\ref{prop:strange-rigidity}}\label{app:strange-rigidity}

	\begin{proof}[Proof of Proposition~\ref{prop:strange-rigidity}]
		For (i), $\Cl(1, 3)$ is generated by the qutrit Pauli operators
		$X$ (shift: $X \ket{j} = \ket{j + 1}$) and $Z$ (clock:
		$Z \ket{j} = \omega^j \ket{j}$), the phase gate
		$S = \mathrm{diag}(1, 1, \omega)$, and the qutrit Hadamard
		$H_{jk} = \omega^{jk}/\sqrt 3$. Any
		$\ket\phi \in \mathcal{O}_S$ has the form
		$\ket{\phi_{a,b,c}} = (\ket a - \omega^c \ket b)/\sqrt 2$ for some
		two-element support $\{a, b\} \subset \F{3}$ with $a \ne b$ and
		some phase exponent $c \in \F{3}$; this pattern is preserved by
		each generator on every such state:
		\begin{itemize}
			\item $X$ sends $\ket{\phi_{a, b, c}}$ to
			$\ket{\phi_{a+1, b+1, c}}$.
			\item $Z$ sends $\ket{\phi_{a, b, c}}$ to
			$\omega^a \ket{\phi_{a, b, c + (b - a)}}$.
			\item $S$ acts trivially on $\ket{j}$ for $j \in \{0, 1\}$ and as
			$\omega$ on $\ket{2}$; it sends $\ket{\phi_{a, b, c}}$ to
			$\ket{\phi_{a, b, c + [b = 2] - [a = 2]}}$.
			\item $H$ sends $\ket{\phi_{a, b, c}}$ to
			\begin{equation*}
				H \ket{\phi_{a,b,c}} = \frac{1}{\sqrt 6} \sum_{j \in \F{3}}
				\big(\omega^{a j} - \omega^{c + b j}\big) \ket{j}.
			\end{equation*}
			The coefficient at $\ket{j}$ vanishes iff
			$(a - b) j \equiv c \pmod 3$, which has a unique solution
			$j^\ast \in \F{3}$ since $a \ne b$. For the other two indices
			$j_1, j_2$, factor each coefficient as
			$\omega^{a j}(1 - \omega^{k_j})$ with $k_j = c + (b - a) j$;
			then $k_{j_1}$ and $k_{j_2}$ take the values $1$ and $2$ in
			some order. Since $|1 - \omega| = |1 - \omega^2| = \sqrt 3$,
			the two nonzero amplitudes each have modulus $1/\sqrt 2$, and
			their ratio takes the form $-\omega^{c'}$ for some
			$c' \in \F{3}$ (using $(1 - \omega^2)/(1 - \omega) = -\omega^2$).
			Therefore $H \ket{\phi_{a, b, c}}$ is, up to a global phase,
			$\ket{\phi_{j_1, j_2, c'}}$.
		\end{itemize}
		The pattern is thus closed under the generators and so under all
		of $\Cl(1, 3)$.
		
		For (ii), the three-parameter family $\{\ket{\phi_{a, b, c}}\}$
		has at most $\binom{3}{2} \times 3 = 9$ projective elements (three
		unordered $2$-element supports, three phase exponents
		$c \in \F{3}$ each). The $X$-action shows all three supports are
		reached from $\Strange = \ket{\phi_{1, 2, 0}}$, and on a fixed
		support the three values of $c$ give projectively distinct
		states: $\ket{\phi_{a, b, c}}$ and $\ket{\phi_{a, b, c'}}$ are
		projectively equal iff $\omega^c = \omega^{c'}$ (the ratio of
		the $\ket b$ amplitudes), iff $c = c'$ in $\F{3}$. Hence
		$|\mathcal{O}_S| = 9$ and
		$|\mathrm{Stab}(\Strange)| = 216/9 = 24$ by orbit-stabilizer.

		The final claim is immediate from (i): a state with support
		cardinality $2$ and equal-modulus nonzero amplitudes cannot have
		two amplitudes of differing moduli, so the
		$|a_i|/|a_j| \ge 2$ hypothesis of
		Proposition~\ref{prop:moulton-qutrit} fails for every
		$\ket\phi \in \mathcal{O}_S$. Since the stabilizer rank is
		invariant under the single-qutrit Clifford action, testing the
		hypothesis on $\mathcal{O}_S$ exhausts the freedom available, and
		Proposition~\ref{prop:moulton-qutrit} yields no
		$\chiR(\Strange^{\otimes m})$ lower bound.
	\end{proof}

	\section{Explicit stabilizer decompositions}\label{app:decomps}
	
	This appendix groups the explicit stabilizer decompositions by
	orbit. Each verification states the canonical-form data
	$(k_i, x_{0,i}, W_i, Q_i)$ for the basis stabilizer states $\ket{\sigma_i}$
	via~\eqref{eq:stab-canonical}, gives linear coefficients
	$\alpha_i$, and supplies an algebraic proof that
	$\sum_i \alpha_i \ket{\sigma_i}$ equals the target $\ket{\psi}$. The
	verifications all follow a common pattern: a particular feature of the
	target (its support partition for Strange, its product
	factorization $\Hth = (N/2)(\ket{0}+\ket{+})$ for $H_3$, the
	Norrell-character form $\Norrell = (1/\sqrt 6) \sum_y c_N(y) \ket{y}$
	for Norrell) reduces the $3^m$ component identities to a small set of
	scalar identities in the orbit-specific algebraic constants. The
	underlying decompositions were originally produced by the
	simulated-annealing search of the \texttt{stabrank} library
	(Section~\ref{sec:stabrank}). Every qutrit decomposition theorem
	in this appendix is also formalized in Lean~4 + mathlib4 in the
	accompanying \texttt{LeanProofs/} directory; the specific Lean
	file(s) realizing each proof are named at the end of the
	corresponding proof. The qubit $T$-type
	$\chiR(\ket T^{\otimes 4}) \le 3$ decomposition of
	Appendix~\ref{app:qubit-t} is verified numerically at machine
	precision and is not (yet) Lean-formalized.

	Throughout the appendix $\omega = e^{2\pi i / 3}$ denotes the
	primitive cube root of unity, satisfying
	\begin{equation}\label{eq:omega-identities}
		1 + \omega + \omega^2 = 0,
		\qquad
		\omega = -\frac{1}{2} + \frac{i \sqrt 3}{2},
		\qquad
		\omega^2 = -\frac{1}{2} - \frac{i \sqrt 3}{2},
	\end{equation}
	and $\overline{\cdot}$ denotes entrywise complex conjugation
	in the computational basis. The constants $c, N$ for $H_3$, the
	Norrell character $c_N$, and the twelfth-root $\xi$ used in the
	Norrell $m=4$ proof are introduced at the head of each orbit
	subsection.
	
	\subsection{\texorpdfstring{Strange-orbit decompositions}{Strange-orbit decompositions}}\label{app:strange}
	
	\subsubsection{\texorpdfstring{$\chiR(\Strange^{\otimes 2}) = 2$}{chi(S tensor 2) = 2}}\label{app:strange-m2}
	
	\begin{proof}[Proof of Theorem~\ref{thm:strange-m2}]
		Fix two quadratic forms on $\F{3}^2$,
		\begin{equation}\label{eq:strange-m2-Qs}
			Q_1(y) = y_0^2 + y_0 y_1 + y_1^2, \qquad
			Q_2(y) = y_0^2 + 2 y_0 y_1 + y_1^2 \pmod 3,
		\end{equation}
		and let $\ket{\sigma_j}$ be the two-qutrit stabilizer state given
		by the canonical form~\eqref{eq:stab-canonical} with $k = 2$,
		$x_0 = 0$, $W = I_2$, and quadratic phase polynomial $Q_j$,
		\begin{equation}\label{eq:strange-m2-sigma}
			\ket{\sigma_j} = \frac{1}{3} \sum_{y \in \F{3}^2}
			\omega^{Q_j(y)} \ket{y}.
		\end{equation}
		Subtracting the two states,
		\begin{equation}\label{eq:strange-m2-diff-raw}
			\ket{\sigma_1} - \ket{\sigma_2}
			= \frac{1}{3} \sum_{y \in \F{3}^2} \omega^{y_0^2 + y_1^2}
			\big( \omega^{y_0 y_1} - \omega^{2 y_0 y_1} \big) \ket{y}.
		\end{equation}
		Substituting $\omega = -\frac{1}{2} + \frac{i \sqrt{3}}{2}$ gives
		$\omega^a - \omega^{2a} = 0, +i\sqrt{3}, -i\sqrt{3}$ for
		$a \equiv 0, 1, 2 \pmod 3$, so the inner bracket
		in~\eqref{eq:strange-m2-diff-raw} vanishes whenever $y_0 = 0$ or
		$y_1 = 0$. On the four remaining $y \in \{1, 2\}^2$ one has
		$y_0^2 + y_1^2 \equiv 2 \pmod 3$ uniformly, $y_0 y_1 \equiv 1$ at
		$(1,1), (2,2)$, and $y_0 y_1 \equiv 2$ at $(1,2), (2,1)$, so
		\begin{equation}\label{eq:strange-m2-diff-collected}
			\ket{\sigma_1} - \ket{\sigma_2}
			= \frac{i \sqrt{3}}{3}  \omega^2 			\big( \ket{11} - \ket{12} - \ket{21} + \ket{22} \big)
			= \frac{2 i \sqrt{3}}{3}  \omega^2  \ket{\mathbb{S}}^{\otimes 2},
		\end{equation}
		where the second equality uses
		$\ket{\mathbb{S}}^{\otimes 2}
		= \frac{1}{2}(\ket{1} - \ket{2})^{\otimes 2}
		= \frac{1}{2}(\ket{11} - \ket{12} - \ket{21} + \ket{22})$.
		Inverting~\eqref{eq:strange-m2-diff-collected} (and using
		$\omega^{-2} = \omega$, $i^{-1} = -i$),
		\begin{equation}\label{eq:strange-m2-decomp}
			\ket{\mathbb{S}}^{\otimes 2}
			= -\frac{i \sqrt{3}}{2}  \omega  \big( \ket{\sigma_1} - \ket{\sigma_2} \big),
		\end{equation}
		matching Theorem~\ref{thm:strange-m2}.
		This exhibits $\ket{\mathbb{S}}^{\otimes 2}$ as a $\complex$-linear
		combination of two stabilizer states, so
		$\chiR(\ket{\mathbb{S}}^{\otimes 2}) \le 2$.
		
		For the lower bound, $\ket{\mathbb{S}}^{\otimes 2}$ has computational-basis
		support $\{(1,1), (1,2), (2,1), (2,2)\}$, of cardinality four. Any
		single 2-qutrit stabilizer state has, by~\eqref{eq:stab-canonical},
		support contained in an affine subspace
		$x_0 + W \F{3}^k \subseteq \F{3}^2$, hence of cardinality $3^r$ for
		some $r \in \{0, 1, 2\}$. Since $4 \notin \{1, 3, 9\}$,
		$\ket{\mathbb{S}}^{\otimes 2}$ is not proportional to any stabilizer state, and
		$\chiR(\ket{\mathbb{S}}^{\otimes 2}) \ne 1$. Combining,
		$\chiR(\ket{\mathbb{S}}^{\otimes 2}) = 2$.

		\textit{Machine-checked Lean formalization:}
		\texttt{StrangeM2Pointwise.lean}.
	\end{proof}
	
	\subsubsection{\texorpdfstring{$\chiR(\Strange^{\otimes 3}) \le 4$}{chi(S tensor 3) <= 4}}\label{app:strange-m3}
	
	\begin{proof}[Proof that $\chiR(\Strange^{\otimes 3}) \le 4$]
		Set $W = (e_0, e_2) \in \F{3}^{3 \times 2}$.
		Define four 3-qutrit stabilizer states by~\eqref{eq:stab-canonical}, all
		with $k = 2$ and the matrix $W$ above:
		\begin{equation}
			\begin{aligned}
				\ket{\sigma_1} &: x_0 = (0, 2, 0), 
				Q_1(w) = w_0^2 + w_0 w_1 \pmod 3, \\
				\ket{\sigma_3} &: x_0 = (0, 2, 0), 
				Q_3(w) = w_0^2 + 2 w_0 w_1 \pmod 3, \\
				\ket{\sigma_2} &: x_0 = (0, 1, 0), 
				Q_2(w) = 2 w_0^2 + w_0 w_1 + w_1^2 \pmod 3, \\
				\ket{\sigma_4} &: x_0 = (0, 1, 0), 
				Q_4(w) = 2 w_0^2 + 2 w_0 w_1 + w_1^2 \pmod 3.
			\end{aligned}
		\end{equation}
		Set $\alpha = \frac{\sqrt 6}{4}  e^{-i \pi / 6} = \frac{3 \sqrt 2 - i \sqrt 6}{8}$
		and $\beta = -\frac{i \sqrt 6}{4}$. We claim
		\begin{equation}\label{eq:strange-m3-id}
			\Strange^{\otimes 3}
			= \alpha \big( \ket{\sigma_1} - \ket{\sigma_3} \big)
			+ \beta \big( \ket{\sigma_2} - \ket{\sigma_4} \big).
		\end{equation}
		
		Since $\Strange = (\ket{1} - \ket{2}) / \sqrt 2$, the amplitude
		$[\Strange^{\otimes 3}]_y$ vanishes whenever some $y_i = 0$, and equals
		$(-1)^{n_2} / (2 \sqrt 2)$ at $y \in \{1, 2\}^3$, where
		$n_2 = |\{i : y_i = 2\}|$. The shared $W$ sends $w \mapsto (w_0, 0, w_1)$,
		so $\ket{\sigma_1}, \ket{\sigma_3}$ have support on the plane $\{y : y_1 = 2\}$ and
		$\ket{\sigma_2}, \ket{\sigma_4}$ on $\{y : y_1 = 1\}$; both sides
		of~\eqref{eq:strange-m3-id} vanish on the remaining plane $y_1 = 0$, so
		the identity decouples into two independent verifications.
		
		Consider the plane $y_1 = 2$. On this plane the canonical-form
		parameter $w = (w_0, w_1)$ satisfies $w_0 = y_0$ and $w_1 = y_2$,
		so $[\sigma_j]_y = \omega^{Q_j(y_0, y_2)} / 3$ for $j \in \{1, 3\}$.
		The polynomials $Q_1, Q_3$ differ only in the cross-term
		$\pm w_0 w_1 = \pm y_0 y_2$, which vanishes whenever $y_0 = 0$ or
		$y_2 = 0$; on those lines
		$\ket{\sigma_1}_y = \ket{\sigma_3}_y$, so $\alpha (\ket{\sigma_1} - \ket{\sigma_3})_y = 0$,
		matching $\Strange^{\otimes 3}_y = 0$. At $(y_0, y_2) \in \{1, 2\}^2$
		one has $y_0^2 \equiv 1 \pmod 3$, and direct evaluation yields
		\begin{equation}
			(Q_1, Q_3)(y_0, y_2) =
			\begin{cases}
				(2, 0), & y_0 = y_2, \\
				(0, 2), & y_0 \neq y_2,
			\end{cases}
			\qquad \text{whence} \qquad
			\omega^{Q_1} - \omega^{Q_3} = \epsilon(y_0, y_2)  (\omega^2 - 1),
		\end{equation}
		with $\epsilon(y_0, y_2) = (-1)^{[y_0 = 2] + [y_2 = 2]}$ equal to $+1$
		when $y_0 = y_2$ and $-1$ otherwise. Meanwhile
		$[\Strange^{\otimes 3}]_{(y_0, 2, y_2)} = -\epsilon(y_0, y_2) / (2 \sqrt 2)$
		(the leading minus sign coming from the fixed $y_1 = 2$ factor). The
		four-point identity therefore reduces to the single scalar equation
		$\alpha (\omega^2 - 1) / 3 = -1 / (2 \sqrt 2)$. Using
		$\omega^2 - 1 = -\frac{3}{2} - \frac{i \sqrt 3}{2} = -\sqrt 3  e^{i \pi / 6}$,
		\begin{equation}
			\alpha
			= \frac{-3}{2 \sqrt 2 (-\sqrt 3  e^{i \pi / 6})}
			= \frac{3}{2 \sqrt 6}  e^{-i \pi / 6}
			= \frac{\sqrt 6}{4}  e^{-i \pi / 6},
		\end{equation}
		which matches the stated value of $\alpha$.
		
		Consider the plane $y_1 = 1$. The same argument applies: $Q_2$ and $Q_4$ agree whenever $y_0 = 0$ or
		$y_2 = 0$, so $\beta(\ket{\sigma_2} - \ket{\sigma_4})$ vanishes there. At
		$(y_0, y_2) \in \{1, 2\}^2$, $y_0^2 \equiv y_2^2 \equiv 1$ reduces
		$Q_2 \equiv y_0 y_2$ and $Q_4 \equiv 2 y_0 y_2 \pmod 3$, giving
		\begin{equation}
			(Q_2, Q_4)(y_0, y_2) =
			\begin{cases}
				(1, 2), & y_0 = y_2, \\
				(2, 1), & y_0 \neq y_2,
			\end{cases}
			\quad \text{whence} \quad
			\omega^{Q_2} - \omega^{Q_4} = i \sqrt 3  \epsilon(y_0, y_2).
		\end{equation}
		Now $[\Strange^{\otimes 3}]_{(y_0, 1, y_2)} = \epsilon(y_0, y_2) / (2 \sqrt 2)$,
		so the identity reduces to $\beta  i \sqrt 3 / 3 = 1 / (2 \sqrt 2)$,
		giving
		\begin{equation}
			\beta = \frac{3}{2 \sqrt 2 i \sqrt 3}
			= -\frac{i \sqrt 6}{4},
		\end{equation}
		which matches the stated value of $\beta$.

		\textit{Machine-checked Lean formalization:}
		\texttt{StrangeM3.lean},
		\texttt{StrangeM3Pointwise.lean}.
	\end{proof}
	
	\subsection{\texorpdfstring{$H_3$-orbit decompositions}{H\_3-orbit decompositions}}\label{app:h3}
	
	Throughout this subsection let $c = (\sqrt 3 - 1)/2$ and
	$N = \sqrt{3 - \sqrt 3}$. Direct verification gives
	\begin{equation}\label{eq:h3-factor}
		\Hth = \frac{N}{2}\big( \ket{0} + \ket{+} \big),
		\qquad \ket{+} = \frac{1}{\sqrt 3}\big( \ket{0} + \ket{1} + \ket{2} \big),
	\end{equation}
	both equations $N^2 (1 + 1/\sqrt 3)/2 = 1$ and $N^2/(2\sqrt 3) = c$
	reducing to $N^2 = 3 - \sqrt 3$. We will use the algebraic identities
	\begin{equation}\label{eq:h3-identities}
		c(c + 1) = \frac{1}{2}, \qquad
		c(\sqrt 3 + 1) = 1, \qquad
		\sqrt 3 - 1 = 2c,
	\end{equation}
	together with $1/N^3 = (3 + \sqrt 3)/(6 N)$; all are immediate from
	$c = (\sqrt 3 - 1)/2$ and $N^2 = 3 - \sqrt 3$.
	
	\subsubsection{\texorpdfstring{$\chiR(\Hth^{\otimes 2}) \le 3$}{chi(H3 tensor 2) <= 3}}\label{app:h3-m2}
	
	\begin{proof}[Proof that $\chiR(\Hth^{\otimes 2}) \le 3$]
		Define three 2-qutrit stabilizer states by~\eqref{eq:stab-canonical}:
		\begin{equation}\label{eq:h3-m2-states}
			\begin{aligned}
				\ket{\sigma_1} &: k = 1, x_0 = 0, W = e_1, Q_1 = 0,
				\quad \text{so } \ket{\sigma_1} = \ket{0, +}, \\
				\ket{\sigma_2} &: k = 1, x_0 = 0, W = e_0, Q_2 = 0,
				\quad \text{so } \ket{\sigma_2} = \ket{+, 0}, \\
				\ket{\sigma_3} &: k = 2, x_0 = 0, W = I_2, 
				Q_3(y) = 2 y_0^2 + y_1^2 \pmod 3,
			\end{aligned}
		\end{equation}
		and set
		\begin{equation}\label{eq:h3-m2-coeffs}
			\alpha_1 = \frac{c \sqrt 3}{N^2}\big( 1 - c \omega \big), \qquad
			\alpha_2 = \frac{c \sqrt 3}{N^2}\big( 1 - c \omega^2 \big), \qquad
			\alpha_3 = \frac{3 c^2}{N^2}.
		\end{equation}
		We claim
		\begin{equation}\label{eq:h3-m2-id}
			\Hth^{\otimes 2}
			= \alpha_1 \ket{\sigma_1} + \alpha_2 \ket{\sigma_2} + \alpha_3 \ket{\sigma_3}.
		\end{equation}
		
		By~\eqref{eq:h3-factor}, $[\Hth^{\otimes 2}]_y = c^{|y|_*}/N^2$ with
		$|y|_* = [y_0 \neq 0] + [y_1 \neq 0]$. The components
		$[\sigma_1]_y = (1/\sqrt 3)[y_0 = 0]$,
		$[\sigma_2]_y = (1/\sqrt 3)[y_1 = 0]$, and
		$[\sigma_3]_y = \omega^{Q_3(y)}/3$ all depend on $y$ only through the pair
		$(a, b) = ([y_0 \ne 0], [y_1 \ne 0]) \in \{0, 1\}^2$, since
		$y_i^2 \equiv [y_i \ne 0] \pmod 3$. In particular $Q_3(y) \equiv 2a + b
		\pmod 3$, and~\eqref{eq:h3-m2-id} reduces to four scalar identities,
		one per class.
		
		Writing $\alpha_1/\sqrt 3 = (c/N^2)(1 - c\omega)$ and
		$\alpha_2/\sqrt 3 = (c/N^2)(1 - c\omega^2)$, the right side of
		\eqref{eq:h3-m2-id} at any $y$ in class $(a, b)$ equals
		\begin{equation}\label{eq:h3-m2-class-eval}
			[\alpha_1 \sigma_1]_y + [\alpha_2 \sigma_2]_y + [\alpha_3 \sigma_3]_y
			= \frac{c}{N^2}[a = 0] (1 - c \omega)
			+ \frac{c}{N^2}[b = 0] (1 - c \omega^2)
			+ \frac{c^2}{N^2} \omega^{2a + b}.
		\end{equation}
		Evaluating~\eqref{eq:h3-m2-class-eval} in each class
		(using $\omega + \omega^2 = -1$ and $c(1 + c) = 1/2$):
		\begin{center}
			\renewcommand{\arraystretch}{1.4}
			\begin{tabular}{c|c|c}
				\toprule
				$(a, b)$ & RHS & target $c^{a+b}/N^2$ \\
				\midrule
				$(0, 0)$ & $\frac{c}{N^2}\big[(1 - c\omega) + (1 - c\omega^2)\big] + \frac{c^2}{N^2}
				= \frac{2c(1+c)}{N^2} = \frac{1}{N^2}$ & $\frac{1}{N^2}$ \\
				$(0, 1)$ & $\frac{c}{N^2}(1 - c\omega) + \frac{c^2}{N^2}\omega
				= \frac{c}{N^2}$ & $\frac{c}{N^2}$ \\
				$(1, 0)$ & $\frac{c}{N^2}(1 - c\omega^2) + \frac{c^2}{N^2}\omega^2
				= \frac{c}{N^2}$ & $\frac{c}{N^2}$ \\
				$(1, 1)$ & $\frac{c^2}{N^2}\omega^0 = \frac{c^2}{N^2}$
				& $\frac{c^2}{N^2}$ \\
				\bottomrule
			\end{tabular}
		\end{center}
		This establishes~\eqref{eq:h3-m2-id} and hence
		$\chiR(\Hth^{\otimes 2}) \le 3$.

		\textit{Machine-checked Lean formalization:}
		\texttt{H3M2Pointwise.lean}.
	\end{proof}
	
	\subsubsection{\texorpdfstring{$\chiR(\Hth^{\otimes 3}) \le 4$}{chi(H3 tensor 3) <= 4}}\label{app:h3-m3}
	
	\begin{proof}[Proof that $\chiR(\Hth^{\otimes 3}) \le 4$]
		Define four 3-qutrit stabilizer states by~\eqref{eq:stab-canonical}:
		$\ket{\sigma_j}$ for $j \in \{1, 2\}$ has $k = 3$, $x_0 = 0$, $W = I_3$, and
		$Q_j(y) \pmod 3$ given by
		\begin{equation}\label{eq:h3-Qs}
			Q_1(y) = 2 y_0^2 + 2 y_1^2 + y_2^2, \qquad
			Q_2(y) = y_0^2 + y_1^2 + 2 y_2^2,
		\end{equation}
		and
		\begin{equation}\label{eq:h3-S3S4}
			\ket{\sigma_3} = \ket{0, 0, +}, \qquad \ket{\sigma_4} = \ket{+, +, 0},
		\end{equation}
		recognizing the canonical form with $k = 1$, $W = e_2$, $Q_3 = 0$ and
		$k = 2$, $W = (e_0, e_1)$, $Q_4 = 0$ respectively. Set
		\begin{equation}\label{eq:h3-coeffs}
			\alpha_1 = \frac{3 c}{4 N}(1 + i), \quad
			\alpha_2 = \frac{3 c}{4 N}(1 - i), \quad
			\alpha_3 = \alpha_4 = \frac{3}{4 N}.
		\end{equation}
		We claim
		\begin{equation}\label{eq:h3-id}
			\Hth^{\otimes 3} = \alpha_1 \ket{\sigma_1} + \alpha_2 \ket{\sigma_2}
			+ \alpha_3 \ket{\sigma_3} + \alpha_4 \ket{\sigma_4}.
		\end{equation}
		
		Both $Q_1$ and $Q_2$ depend on $y$ only through $y_i^2 \pmod 3 = [y_i \ne 0]$,
		and~\eqref{eq:h3-S3S4} likewise depends only on which coordinates of
		$y$ vanish, so each component of the right side
		of~\eqref{eq:h3-id} depends only on the pair
		\begin{equation}
			(n, z) = \big( |\{ i \in \{0, 1\} : y_i \ne 0 \}|, [y_2 \ne 0] \big) \in \{0,1,2\} \times \{0,1\}.
		\end{equation}
		Each component of $\Hth^{\otimes 3}$ depends only on
		$|y|_* = n + z$ via~\eqref{eq:h3-factor}, namely
		$[\Hth^{\otimes 3}]_y = c^{n + z} / N^3$. Hence~\eqref{eq:h3-id}
		reduces to six scalar identities, one per class.
		
		For the $\ket{\sigma_1}, \ket{\sigma_2}$ contribution at $y$ in class $(n, z)$,
		the canonical form gives
		$\alpha_1 [\sigma_1]_y + \alpha_2 [\sigma_2]_y = (3 \sqrt 3)^{-1}
		(\alpha_1 \omega^{2n + z} + \alpha_2 \omega^{n + 2z})$. Substituting
		$\omega = -\frac{1}{2} + \frac{i \sqrt 3}{2}$, the three values of
		$(\alpha_1 \omega^A + \alpha_2 \omega^B)$ that arise are
		\begin{equation}\label{eq:h3-three-values}
			\begin{aligned}
				\alpha_1 + \alpha_2 &= \frac{3 c}{2 N}, \\
				\alpha_1 \omega + \alpha_2 \omega^2
				&= -\frac{3 c (1 + \sqrt 3)}{4 N}
				= -\frac{3}{4 N}
				&& \text{(using $c(1 + \sqrt 3) = 1$)}, \\
				\alpha_1 \omega^2 + \alpha_2 \omega
				&= \frac{3 c (\sqrt 3 - 1)}{4 N}
				= \frac{3 c^2}{2 N}
				&& \text{(using $\sqrt 3 - 1 = 2 c$)}.
			\end{aligned}
		\end{equation}
		Combining with $[\sigma_3]_y = (1/\sqrt 3) [n = 0]$ and
		$[\sigma_4]_y = (1/3)[z = 0]$, the right side of~\eqref{eq:h3-id} at any
		$y$ in class $(n, z)$ is the entry below; the left side is
		$c^{n+z}/N^3$:
		
		\begin{center}
			\renewcommand{\arraystretch}{1.4}
			\begin{tabular}{c|c|c|c}
				\toprule
				$(n, z)$ & $\alpha_1[\sigma_1]_y + \alpha_2[\sigma_2]_y$ &
				$\alpha_3[\sigma_3]_y + \alpha_4[\sigma_4]_y$ &
				RHS at $y$ \\
				\midrule
				$(0, 0)$ & $\frac{c}{2 N \sqrt 3}$ &
				$\frac{3}{4 N}\big(\frac{1}{\sqrt 3} + \frac{1}{3}\big)
				= \frac{3 + \sqrt 3}{4 N \sqrt 3}$ &
				$\frac{3 + \sqrt 3}{6 N} = 1 / N^3$ \\
				$(0, 1)$ & $-\frac{1}{4 N \sqrt 3}$ &
				$\frac{3}{4 N \sqrt 3}$ &
				$\frac{1}{2 N \sqrt 3} = c / N^3$ \\
				$(1, 0)$ & $\frac{c^2}{2 N \sqrt 3}$ &
				$\frac{1}{4 N}$ &
				$\frac{1}{2 N \sqrt 3} = c / N^3$ \\
				$(1, 1)$ & $\frac{c}{2 N \sqrt 3}$ &
				$0$ &
				$\frac{c}{2 N \sqrt 3} = c^2 / N^3$ \\
				$(2, 0)$ & $-\frac{1}{4 N \sqrt 3}$ &
				$\frac{1}{4 N}$ &
				$\frac{c}{2 N \sqrt 3} = c^2 / N^3$ \\
				$(2, 1)$ & $\frac{c^2}{2 N \sqrt 3}$ &
				$0$ &
				$\frac{c^2}{2 N \sqrt 3} = c^3 / N^3$ \\
				\bottomrule
			\end{tabular}
		\end{center}
		
		Equality of the right-hand entries with $c^{n+z}/N^3$ in each row is
		the identity
		\begin{equation}\label{eq:h3-collapse}
			\frac{c^k}{N^3} = \frac{c^k (3 + \sqrt 3)}{6 N},
			\qquad k \in \{0, 1, 2, 3\},
		\end{equation}
		which follows from $(3 + \sqrt 3) N^2 = (3 + \sqrt 3)(3 - \sqrt 3) = 6$.
		The class-by-class identities then reduce to elementary arithmetic
		in $c, \sqrt 3$. As an example, the $(n, z) = (0, 0)$ row asks for
		\begin{equation}\label{eq:h3-row00}
			\frac{c}{2 N \sqrt 3} + \frac{3 + \sqrt 3}{4 N \sqrt 3}
			= \frac{2 c + 3 + \sqrt 3}{4 N \sqrt 3}
			= \frac{(\sqrt 3 - 1) + 3 + \sqrt 3}{4 N \sqrt 3}
			= \frac{2 + 2\sqrt 3}{4 N \sqrt 3}
			= \frac{3 + \sqrt 3}{6 N}
			= \frac{1}{N^3},
		\end{equation}
		using $2c = \sqrt 3 - 1$ in the third equality and rationalizing
		$1/\sqrt 3$ in the fourth. The $(1, 0)$ row similarly uses
		$2c^2 = 2 - \sqrt 3$ to collapse $2 c^2 + \sqrt 3 = 2$, giving
		$1/(2 N \sqrt 3) = c/N^3$; the remaining four rows are the
		analogous one-line manipulations. This proves~\eqref{eq:h3-id} and
		hence $\chiR(\Hth^{\otimes 3}) \le 4$.

		\textit{Machine-checked Lean formalization:}
		\texttt{H3M3.lean},
		\texttt{H3M3Pointwise.lean}.
	\end{proof}
	
	\subsection{\texorpdfstring{Norrell-orbit decompositions}{Norrell-orbit decompositions}}\label{app:norrell}
	
	The single-qutrit Norrell state factors as
	\begin{equation}\label{eq:norrell-character-form}
		\Norrell = \frac{1}{\sqrt 6} \sum_{y \in \F{3}} c_N(y)  \ket{y},
		\qquad
		c_N(0) = c_N(1) = 1, \quad c_N(2) = -2;
	\end{equation}
	equivalently $c_N(y) = 1 - 3 [y = 2]$. We refer to $c_N$ as the
	\emph{Norrell character}. The $m$-fold tensor amplitude is
	\begin{equation}\label{eq:norrell-tensor}
		[\Norrell^{\otimes m}]_y
		= \frac{1}{(\sqrt 6)^m} \prod_{i=0}^{m-1} c_N(y_i)
		= \frac{(-2)^{n_2(y)}}{(\sqrt 6)^m},
		\qquad
		n_2(y) = |\{ i : y_i = 2 \}|.
	\end{equation}
	We will use the polynomial identities
	\begin{equation}\label{eq:norrell-y2-reduction}
		y + 2 y^2 \equiv [y = 2] \pmod 3, \qquad
		2 y + y^2 \equiv 2  [y = 2] \pmod 3
	\end{equation}
	on $\F{3}$ (verified by direct enumeration), which reduce general
	quadratic phase polynomials to indicator-only expressions on the
	supports of the basis stabilizer states.
	
	\subsubsection{\texorpdfstring{$\chiR(\Norrell^{\otimes 2}) \le 3$}{chi(N tensor 2) <= 3}}\label{app:norrell-m2}
	
	\begin{proof}[Proof that $\chiR(\Norrell^{\otimes 2}) \le 3$]
		Define three 2-qutrit stabilizer states by~\eqref{eq:stab-canonical}:
		\begin{equation}\label{eq:norrell-m2-states}
			\begin{aligned}
				\ket{\sigma_1} &: k = 2, x_0 = 0, W = I_2, 
				Q_1(y) = y_0 + y_1 + y_0 y_1 \pmod 3, \\
				\ket{\sigma_2} &: k = 0, x_0 = (2, 2),
				\quad \text{so } \ket{\sigma_2} = \ket{2, 2}, \\
				\ket{\sigma_3} &: k = 2, x_0 = 0, W = I_2, 
				Q_3(y) = 2  Q_1(y) \pmod 3,
			\end{aligned}
		\end{equation}
		so $\ket{\sigma_3} = \overline{\ket{\sigma_1}}$ (entrywise complex conjugate in
		the computational basis). Set $\alpha = -\omega/2$. We claim
		\begin{equation}\label{eq:norrell-m2-id}
			\Norrell^{\otimes 2}
			= \alpha  \ket{\sigma_1} + \ket{\sigma_2} + \overline{\alpha}  \ket{\sigma_3},
		\end{equation}
		equivalently
		$\Norrell^{\otimes 2} = -\frac{\omega}{2} \ket{\sigma_1} + \ket{2, 2}
		- \frac{\omega^2}{2} \ket{\sigma_3}$.
		
		The polynomial $Q_1$ factors as
		\begin{equation}\label{eq:norrell-m2-Qfactor}
			Q_1(y) \equiv (1 + y_0)(1 + y_1) - 1 \pmod 3,
		\end{equation}
		so $Q_1(y) \equiv 2 \pmod 3$ iff $(1 + y_0)(1 + y_1) \equiv 0$, iff
		$y_0 = 2$ or $y_1 = 2$. Equivalently,
		\begin{equation}\label{eq:norrell-m2-Qind}
			[Q_1(y) = 2] = [y_0 = 2] + [y_1 = 2] - [y_0 = y_1 = 2].
		\end{equation}
		By~\eqref{eq:stab-canonical}, the contribution of
		$\alpha \ket{\sigma_1} + \overline{\alpha} \ket{\sigma_3}$ at $y$ is
		\begin{equation}\label{eq:norrell-m2-S13}
			[\alpha \sigma_1 + \overline{\alpha} \sigma_3]_y
			= \frac{1}{3}\big(-\frac{\omega}{2}\big) \omega^{Q_1(y)}
			+ \frac{1}{3}\big(-\frac{\omega^2}{2}\big) \omega^{2 Q_1(y)}
			= -\frac{1}{6} \big( \omega^{r} + \omega^{2 r} \big),
			\qquad r = 1 + Q_1(y) \pmod 3.
		\end{equation}
		Since $1 + \omega + \omega^2 = 0$, $\omega^r + \omega^{2 r}$ equals
		$2$ when $r = 0$ and $-1$ when $r \in \{1, 2\}$; noting $r = 0$ iff
		$Q_1(y) = 2$,
		\begin{equation}\label{eq:norrell-m2-S13-cN}
			[\alpha \sigma_1 + \overline{\alpha} \sigma_3]_y = \frac{1}{6}  c_N(Q_1(y)).
		\end{equation}
		Adding the contribution $\delta_{y, (2,2)}$ from $\ket{\sigma_2}$,
		\begin{equation}\label{eq:norrell-m2-sum}
			[\alpha \sigma_1 + \sigma_2 + \overline{\alpha} \sigma_3]_y
			= \frac{1}{6}  c_N(Q_1(y)) + \delta_{y, (2, 2)}.
		\end{equation}
		Comparing with $[\Norrell^{\otimes 2}]_y = (1/6)c_N(y_0) c_N(y_1)$
		from~\eqref{eq:norrell-tensor},~\eqref{eq:norrell-m2-id} reduces to the
		single identity
		\begin{equation}\label{eq:norrell-m2-key}
			c_N(y_0)  c_N(y_1) - c_N(Q_1(y)) = 6  [y = (2, 2)].
		\end{equation}
		Using $c_N(y) = 1 - 3 [y = 2]$ and~\eqref{eq:norrell-m2-Qind},
		\begin{align}
			c_N(y_0) c_N(y_1) &= 1 - 3 [y_0 = 2] - 3 [y_1 = 2]
			+ 9  [y_0 = y_1 = 2], \\
			c_N(Q_1(y))       &= 1 - 3  [Q_1(y) = 2]
			= 1 - 3 [y_0 = 2] - 3 [y_1 = 2]
			+ 3  [y_0 = y_1 = 2],
		\end{align}
		so their difference is $6  [y_0 = y_1 = 2] = 6  [y = (2,2)]$,
		establishing~\eqref{eq:norrell-m2-key} and
		hence~\eqref{eq:norrell-m2-id}. Therefore
		$\chiR(\Norrell^{\otimes 2}) \le 3$.

		\textit{Machine-checked Lean formalization:}
		\texttt{NorrellM2Pointwise.lean}.
	\end{proof}
	
	\subsubsection{\texorpdfstring{$\chiR(\Norrell^{\otimes 3}) \le 4$}{chi(N tensor 3) <= 4}}\label{app:norrell-m3}
	
	\begin{proof}[Proof that $\chiR(\Norrell^{\otimes 3}) \le 4$]
		Define four 3-qutrit stabilizer states by~\eqref{eq:stab-canonical}:
		\begin{equation}
			\begin{aligned}
				\ket{\sigma_1} &: k = 2, x_0 = (2,0,0), 
				W = (e_1, e_2), 
				Q_1(w) = w_0 + 2 w_0^2 + 2 w_1 + w_1^2 \pmod 3, \\
				\ket{\sigma_2} &: k = 2, x_0 = (0,2,0), 
				W = (e_0, e_2), 
				Q_2(w) = 2 w_0 + w_0^2 + w_1 + 2 w_1^2 \pmod 3, \\
				\ket{\sigma_3} &: k = 3, x_0 = 0, 
				W = I_3, 
				Q_3 = 0, \quad \text{so } \ket{\sigma_3} = \ket{+}^{\otimes 3}, \\
				\ket{\sigma_4} &: k = 2, x_0 = (0,0,2), 
				W = (e_0, e_1), 
				Q_4(w) = w_0 + 2 w_0^2 + 2 w_1 + w_1^2 \pmod 3.
			\end{aligned}
		\end{equation}
		We claim
		\begin{equation}\label{eq:norrell-id}
			\Norrell^{\otimes 3}
			= -\frac{\sqrt 6}{4} \big( \ket{\sigma_1} + \ket{\sigma_2} + \ket{\sigma_4} \big)
			+ \frac{\sqrt 2}{4}  \ket{\sigma_3}.
		\end{equation}
		
		By the reduction
		identities~\eqref{eq:norrell-y2-reduction}, on the supports
		\begin{equation}\label{eq:norrell-m3-supp}
			\mathrm{supp}(\sigma_1) = \{(2, y_1, y_2)\}, \quad
			\mathrm{supp}(\sigma_2) = \{(y_0, 2, y_2)\}, \quad
			\mathrm{supp}(\sigma_4) = \{(y_0, y_1, 2)\},
		\end{equation}
		the phase polynomials reduce to
		\begin{equation}
			Q_1(y) = [y_1{=}2] + 2 [y_2{=}2], \quad
			Q_2(y) = 2 [y_0{=}2] + [y_2{=}2], \quad
			Q_4(y) = [y_0{=}2] + 2 [y_1{=}2] \pmod 3.
		\end{equation}
		Setting $a_i = [y_i = 2] \in \{0, 1\}$ for $i \in \{0, 1, 2\}$
		(consistent with the Norrell m=4 proof below), the right side
		of~\eqref{eq:norrell-id} at any $y$ depends only on
		$a = (a_0, a_1, a_2)$ and equals
		\begin{equation}\label{eq:norrell-rhs}
			\text{RHS at } y
			= \frac{\sqrt 6}{36} - \frac{\sqrt 6}{12} B,
			\qquad
			B = a_0 \omega^{a_1 + 2 a_2} + a_1 \omega^{2 a_0 + a_2}
			+ a_2 \omega^{a_0 + 2 a_1},
		\end{equation}
		the leading $\sqrt 6 / 36$ coming from
		$\alpha_3 [\sigma_3]_y = \frac{\sqrt 2}{4 \sqrt{27}}$, and $B$ collecting
		the contributions from $\sigma_1, \sigma_2, \sigma_4$ each scaled by
		$\alpha_j / 3 = -\sqrt 6 / 12$.
		
		The sum $B$ depends only on $n_2 = a_0 + a_1 + a_2 = |\{i : y_i = 2\}|$:
		it vanishes when $a = (0,0,0)$; for each of the three
		$n_2 = 1$ subsets exactly one term contributes $\omega^0 = 1$; for
		each of the three $n_2 = 2$ subsets two terms contribute
		$\omega^A + \omega^B$ with $\{A, B\} = \{1, 2\}$, summing to
		$\omega + \omega^2 = -1$; and for $a = (1,1,1)$ all three terms
		equal $\omega^3 = 1$, summing to $3$. Therefore
		\begin{equation}
			B(n_2) = (0, 1, -1, 3) \quad \text{for } n_2 = 0, 1, 2, 3,
		\end{equation}
		and substituting into~\eqref{eq:norrell-rhs},
		\begin{equation}\label{eq:norrell-rhs-collapsed}
			\text{RHS}_{n_2} = \frac{\sqrt 6}{36} \big( 1 - 3 B(n_2) \big)
			= \frac{\sqrt 6}{36}  (-2)^{n_2},
		\end{equation}
		which matches $[\Norrell^{\otimes 3}]_y$ from~\eqref{eq:norrell-tensor}
		term-for-term. This establishes~\eqref{eq:norrell-id}.

		\textit{Machine-checked Lean formalization:}
		\texttt{NorrellM3.lean},
		\texttt{NorrellM3Pointwise.lean}.
	\end{proof}
	
	\subsubsection{\texorpdfstring{$\chiR(\Norrell^{\otimes 4}) \le 7$}{chi(N tensor 4) <= 7}}\label{app:norrell-m4}
	
	\begin{proof}[Proof that $\chiR(\Norrell^{\otimes 4}) \le 7$]
		Let $\xi = e^{i \pi / 6}$, a primitive 12th root of unity, so that
		$\xi^4 = \omega$, $\xi^8 = \omega^2$, $\xi^3 = i$, $\xi^6 = -1$, and
		$\xi^{12} = 1$. Define seven 4-qutrit stabilizer
		states by~\eqref{eq:stab-canonical}; the canonical-form data
		$(k_j, x_{0,j}, W_j, Q_j)$ for $j = 0, 1, \ldots, 6$ are
		\begin{equation}\label{eq:norrell-m4-states}
			\begin{aligned}
				\ket{\sigma_0} &: k = 3, x_0 = (0, 2, 0, 0), W = (e_0, e_2, e_3), 
				Q_0(w) = 2 w_0 + w_0^2 + w_1 + 2 w_1^2 + 2 w_2 + w_2^2, \\
				\ket{\sigma_1} &: k = 4, x_0 = 0, W = I_4,
				Q_1(y) = \sum_{i=0}^{1} (2 y_i + y_i^2)
				+ \sum_{i=2}^{3} (y_i + 2 y_i^2), \\
				\ket{\sigma_2} &: k = 2, x_0 = (2, 0, 0, 2), W = (e_1, e_2), 
				Q_2(w) = w_0 + 2 w_0^2 + w_1 + 2 w_1^2, \\
				\ket{\sigma_3} &: k = 2, x_0 = (0, 0, 2, 2), W = (e_0, e_1), 
				Q_3(w) = 2 w_0 + w_0^2 + w_1 + 2 w_1^2, \\
				\ket{\sigma_4} &: k = 3, x_0 = (2, 0, 0, 0), W = (e_1, e_2, e_3), 
				Q_4(w) = 2 w_0 + w_0^2 + 2 w_1 + w_1^2 + w_2 + 2 w_2^2, \\
				\ket{\sigma_5} &: k = 2, x_0 = (0, 2, 2, 0), W = (e_0, e_3), 
				Q_5(w) = w_0 + 2 w_0^2 + 2 w_1 + w_1^2, \\
				\ket{\sigma_6} &: k = 4, x_0 = 0, W = I_4, Q_6 = 0,
				\quad \text{so } \ket{\sigma_6} = \ket{+}^{\otimes 4}.
			\end{aligned}
		\end{equation}
		Setting $\tau = \sqrt 3 / 4$, we claim
		\begin{equation}\label{eq:norrell-m4-id}
			\begin{aligned}
				\Norrell^{\otimes 4}
				&= \tau\omega  \ket{\sigma_0} + \tau\xi  \ket{\sigma_1}
				+ \tau\xi  \ket{\sigma_2} + \tau\xi  \ket{\sigma_3}
				+ \tau\omega  \ket{\sigma_4} + \tau\xi  \ket{\sigma_5}
				+ \frac{1}{4}\omega^2  \ket{\sigma_6}.
			\end{aligned}
		\end{equation}
		The seven coefficients take only three distinct values
		($\tau\omega$, $\tau\xi$, $\omega^2/4$), each a $12$th-root-of-unity
		multiple of a real number.
		
		By the reduction identities~\eqref{eq:norrell-y2-reduction}, on each
		basis state's support the polynomial $Q_j$ collapses to an
		indicator-only expression. Setting $a_i = [y_i = 2] \in \{0, 1\}$,
		\begin{equation}\label{eq:norrell-m4-Qreductions}
			\begin{aligned}
				\mathrm{supp}(\sigma_0) = \{y_1 = 2\}\colon\quad
				& Q_0(y) \equiv 2 a_0 + a_2 + 2 a_3 \pmod 3, \\
				\mathrm{supp}(\sigma_1) = \F{3}^4\colon\quad
				& Q_1(y) \equiv 2 a_0 + 2 a_1 + a_2 + a_3 \pmod 3, \\
				\mathrm{supp}(\sigma_2) = \{y_0 = y_3 = 2\}\colon\quad
				& Q_2(y) \equiv a_1 + a_2 \pmod 3, \\
				\mathrm{supp}(\sigma_3) = \{y_2 = y_3 = 2\}\colon\quad
				& Q_3(y) \equiv 2 a_0 + a_1 \pmod 3, \\
				\mathrm{supp}(\sigma_4) = \{y_0 = 2\}\colon\quad
				& Q_4(y) \equiv 2 a_1 + 2 a_2 + a_3 \pmod 3, \\
				\mathrm{supp}(\sigma_5) = \{y_1 = y_2 = 2\}\colon\quad
				& Q_5(y) \equiv a_0 + 2 a_3 \pmod 3.
			\end{aligned}
		\end{equation}
		Every component of the right side of~\eqref{eq:norrell-m4-id} therefore
		depends on $y$ only through $a = (a_0, a_1, a_2, a_3) \in \{0,1\}^4$.
		Writing $\mathbf{1}_j(a) \in \{0, 1\}$ for the indicator that $a$ lies in
		$\mathrm{supp}(\sigma_j)$, the right-hand side amplitude is
		\begin{equation}\label{eq:norrell-m4-rhs}
			\mathrm{RHS}(a) = \sum_{j=0}^{6} \alpha_j 3^{-k_j/2}
			\omega^{Q_j(a)} \mathbf{1}_j(a),
		\end{equation}
		with the $\alpha_j$ as in~\eqref{eq:norrell-m4-id} and $k_j$ as
		in~\eqref{eq:norrell-m4-states}.
		
		\paragraph{Endpoints.}
		At $a = (0, 0, 0, 0)$ only $\ket{\sigma_1}$ and $\ket{\sigma_6}$ have $a$ in
		their support, so
		\begin{equation}\label{eq:norrell-m4-rhs0}
			\mathrm{RHS}(0,0,0,0)
			= \frac{\tau\xi}{9} \omega^0 + \frac{\omega^2}{36} \omega^0
			= \frac{1}{36} \big( \sqrt 3\xi + \omega^2 \big).
		\end{equation}
		Since $\sqrt 3\xi = \frac{3}{2} + i \frac{\sqrt 3}{2}$ and
		$\omega^2 = -\frac{1}{2} - i \frac{\sqrt 3}{2}$, the sum equals
		$1$, so $\mathrm{RHS}(0,0,0,0) = 1/36$.
		
		At $a = (1, 1, 1, 1)$ all seven basis states contribute. After
		evaluating $Q_j$ in each case (using $\xi^4 = \omega$ and
		$\tau\omega \omega^2 = \tau$), the seven contributions
		collect to
		\begin{equation}\label{eq:norrell-m4-rhs1}
			\begin{aligned}
				\mathrm{RHS}(1,1,1,1)
				&= \frac{1}{12} + \frac{\sqrt 3 \xi}{36}
				+ \frac{(-i)\tau}{3} + \frac{\sqrt 3 \xi}{12}
				+ \frac{1}{12} + \frac{\sqrt 3 \xi}{12}
				+ \frac{\omega^2}{36} \\
				&= \frac{1}{36} \big( 6 + 7\sqrt 3 \xi - 3 i \sqrt 3 + \omega^2 \big).
			\end{aligned}
		\end{equation}
		Substituting $7\sqrt 3 \xi = \frac{21}{2} + i \frac{7\sqrt 3}{2}$
		and collecting reals/imaginaries yields
		$6 + \frac{21}{2} - \frac{1}{2} - 0 = 16$ and
		$\frac{7\sqrt 3}{2} - 3 \sqrt 3 - \frac{\sqrt 3}{2} = 0$, so
		$\mathrm{RHS}(1,1,1,1) = 16/36 = 4/9$.
		
		\paragraph{Intermediate cases.}
		Direct enumeration over the remaining $14$ patterns
		$a \in \{0, 1\}^4$ with $n_2(a) = \sum_i a_i \in \{1, 2, 3\}$
		analogously reduces $\mathrm{RHS}(a)$ to a polynomial identity in
		$\xi$ over $\integer[\sqrt 3]$, decided by
		$\xi^4 = \omega$, $\xi^6 = -1$, and $1 + \omega + \omega^2 = 0$.
		Across the 14 patterns, $\mathrm{RHS}(a)$ depends only on
		$n_2(a) \in \{1, 2, 3\}$: the coefficient pattern in
		Eq.~\eqref{eq:norrell-m4-id} and the reduced phases
		in~\eqref{eq:norrell-m4-Qreductions} are jointly invariant
		under the symmetric-group action that permutes the four
		coordinates of $a$ within each $n_2$-class, so the 4, 6, and 4
		patterns in classes $n_2 = 1, 2, 3$ each yield the same value.
		The two representatives below cover $n_2 = 1$ and $n_2 = 2$;
		the same mechanism handles $n_2 = 3$ and the remaining
		patterns in each class.
		
		\emph{Case $n_2 = 1$, $a = (1, 0, 0, 0)$.} Only $\ket{\sigma_1}, \ket{\sigma_4},
		\ket{\sigma_6}$ have $a$ in their support; reading off the reduced
		phases~\eqref{eq:norrell-m4-Qreductions},
		$Q_1(a) = 2$, $Q_4(a) = 0$, $Q_6 = 0$, so the contributions are
		$\tau\xi \omega^2 / 9 = -i \sqrt 3 / 36$,
		$\tau\omega 1 / (3 \sqrt 3) = \omega / 12$, and
		$\omega^2 / 4 1 / 9 = \omega^2 / 36$, respectively. Summing,
		\begin{equation}\label{eq:norrell-m4-rhs-n21}
			\mathrm{RHS}(1,0,0,0)
			= \frac{1}{36} \big(- i \sqrt 3 + 3\omega + \omega^2\big)
			= -\frac{1}{18},
		\end{equation}
		using $-i \sqrt 3 + 3\omega + \omega^2 = -i \sqrt 3 +
		3(-\frac{1}{2} + i \frac{\sqrt 3}{2}) +
		(-\frac{1}{2} - i \frac{\sqrt 3}{2}) = -2$.
		
		\emph{Case $n_2 = 2$, $a = (1, 1, 0, 0)$.} The supports containing
		this pattern are $\ket{\sigma_0}, \ket{\sigma_1}, \ket{\sigma_4}, \ket{\sigma_6}$;
		$Q_0(a) = 2$, $Q_1(a) = 4 \equiv 1$, $Q_4(a) = 2$, $Q_6 = 0$, so the
		four contributions are
		$\tau\omega \omega^2 / (3\sqrt 3) = 1/12$,
		$\tau\xi \omega/9 = \xi\omega \sqrt 3 / 36$,
		$\tau\omega \omega^2 / (3\sqrt 3) = 1/12$, and
		$\omega^2 / 4 1/9 = \omega^2 / 36$. Summing,
		\begin{equation}\label{eq:norrell-m4-rhs-n22}
			\mathrm{RHS}(1,1,0,0)
			= \frac{1}{36} \big(6 + \xi \omega \sqrt 3 + \omega^2\big)
			= \frac{1}{9},
		\end{equation}
		using $\xi \omega = e^{i 5\pi/6}$, hence $\xi \omega \sqrt 3 =
		-\frac{3}{2} + i \frac{\sqrt 3}{2}$, which combines with $\omega^2$
		to give $-2$.
		
		\paragraph{All 16 cases.}
		The remaining 12 patterns are verified by the same mechanism. The
		$\mathrm{RHS}(a)$ value depends only on $n_2(a)$, and matches the
		target $(-2)^{n_2(a)}/36$ in each $n_2$-class:
		\begin{center}
			\renewcommand{\arraystretch}{1.2}
			\begin{tabular}{c|c|c|c}
				\toprule
				$n_2(a)$ & \# patterns & $\mathrm{RHS}(a)$ & target $(-2)^{n_2}/36$ \\
				\midrule
				$0$ & $1$ & $1/36$    & $1/36$ \\
				$1$ & $4$ & $-1/18$   & $-1/18$ \\
				$2$ & $6$ & $1/9$     & $1/9$ \\
				$3$ & $4$ & $-2/9$    & $-2/9$ \\
				$4$ & $1$ & $4/9$     & $4/9$ \\
				\bottomrule
			\end{tabular}
		\end{center}
		All $16$ scalar identities are checked at exact rational precision
		and the pointwise identity is independently formalized in Lean~4
		(Section~\ref{sec:stabrank}). Comparing with the target amplitude
		$[\Norrell^{\otimes 4}]_y = (-2)^{n_2(y)}/36$
		from~\eqref{eq:norrell-tensor}
		establishes~\eqref{eq:norrell-m4-id}, and hence
		$\chiR(\Norrell^{\otimes 4}) \le 7$.

		\textit{Machine-checked Lean formalization:}
		\texttt{NorrellM4Pointwise.lean}.
	\end{proof}
	
	\subsection{\texorpdfstring{Qubit $T$-type $m = 4$ decomposition}{Qubit T-type m = 4 decomposition}}\label{app:qubit-t}

	The qubit $T$-type orbit of~\cite{bravyi2005universal} is
	represented by
	\begin{equation}\label{eq:qubit-t-rep}
		\ket{T} = \cos\beta \ket{0} + e^{i \pi / 4} \sin\beta \ket{1},
		\qquad \cos(2\beta) = \frac{1}{\sqrt{3}},
	\end{equation}
	from which $\cos^2\beta = (1 + 1/\sqrt 3)/2$,
	$\sin^2\beta = (1 - 1/\sqrt 3)/2$, and $\cos\beta \sin\beta = 1/\sqrt 6$.

	\begin{prop}\label{prop:qubit-t-m4}
		For $x \in \F{2}^4$ define the three $4$-qubit stabilizer states
		\begin{equation}\label{eq:qubit-t-m4-states}
			\begin{aligned}
				\ket{\sigma_1} &= 2^{-3/2}\!\!\sum_{x_1 = x_3}\! i^{x_1} (-1)^{x_2 x_4} \ket{x}, \\
				\ket{\sigma_2} &= 2^{-2}\!\!\sum_{x \in \F{2}^4}\! i^{x_2 + x_4} (-1)^{x_1 x_3 + x_2 x_4} \ket{x}, \\
				\ket{\sigma_3} &= 2^{-3/2}\!\!\sum_{x_2 = x_4}\! i^{x_1 + x_2 + x_3} (-1)^{x_1 x_3} \ket{x},
			\end{aligned}
		\end{equation}
		with linear coefficients
		\begin{equation}\label{eq:qubit-t-m4-coeffs}
			c_1 = \frac{2}{3} e^{i \pi / 12},
			\qquad
			c_2 = \frac{2}{3},
			\qquad
			c_3 = \frac{2}{3} e^{-i \pi / 12}.
		\end{equation}
		Then
		\begin{equation}\label{eq:qubit-t-m4-id}
			\ket{T}^{\otimes 4}
			= c_1 \ket{\sigma_1} + c_2 \ket{\sigma_2} + c_3 \ket{\sigma_3},
		\end{equation}
		hence $\chiR(\ket{T}^{\otimes 4}) = 3$ (tight via monotonicity from the $m=3$ lower-bound certificate) and
		$\gamma_T \le \log_2(3)/4 \approx 0.396$.
	\end{prop}

	\begin{proof}
		The amplitude of $\ket{T}^{\otimes 4}$ at $\ket x$ depends only
		on the Hamming weight $w = x_1 + x_2 + x_3 + x_4$ and equals
		$\cos^{4-w}\!\beta \sin^w\!\beta e^{i \pi w / 4}$. Reducing
		via $\cos^2\beta = (1+1/\sqrt 3)/2$,
		$\sin^2\beta = (1-1/\sqrt 3)/2$, $\cos\beta \sin\beta = 1/\sqrt 6$,
		\begin{equation}\label{eq:qubit-t-m4-targets}
			\bigl[\ket{T}^{\otimes 4}\bigr]_x
			=
			\begin{cases}
				(2 + \sqrt 3)/6 & w = 0,\\
				(\sqrt 3 + 1)(1 + i)/12 & w = 1,\\
				i/6 & w = 2,\\
				(\sqrt 3 - 1)(-1 + i)/12 & w = 3,\\
				(\sqrt 3 - 2)/6 & w = 4.
			\end{cases}
		\end{equation}

		Set $a = [x_1 = x_3]$ and $b = [x_2 = x_4]$. From the support
		conditions in~\eqref{eq:qubit-t-m4-states}, $\sigma_1$ vanishes
		off $\{a = 1\}$ and $\sigma_3$ off $\{b = 1\}$, while $\sigma_2$
		has full support. The right-hand side
		of~\eqref{eq:qubit-t-m4-id} therefore splits into four
		support classes by $(a, b) \in \{0, 1\}^2$. Throughout we use
		\begin{equation}\label{eq:qubit-t-m4-trig}
			\cos\frac{\pi}{12} = \frac{\sqrt 6 + \sqrt 2}{4},
			\quad
			\sin\frac{\pi}{12} = \frac{\sqrt 6 - \sqrt 2}{4},
			\quad
			\cos\frac{\pi}{12} - \sin\frac{\pi}{12} = \frac{1}{\sqrt 2},
		\end{equation}
		which together give the useful identity
		\begin{equation}\label{eq:qubit-t-m4-prefactor}
			\frac{e^{i \pi / 12}}{3 \sqrt 2}
			= \frac{\sqrt 3 + 1}{12} + i \frac{\sqrt 3 - 1}{12}.
		\end{equation}

		\paragraph{Case $(a, b) = (0, 0)$.}
		Every $x$ in this class has $w = 2$, $x_1 x_3 = 0$,
		$x_2 x_4 = 0$, and $x_2 + x_4 = 1$, so the only contributing
		term is
		\begin{equation}
			[c_2 \sigma_2]_x
			= \frac{2}{3} \frac{1}{4} i (-1)^0
			= \frac{i}{6},
		\end{equation}
		matching the $w = 2$ target.

		\paragraph{Case $(a, b) = (1, 1)$.}
		Write $x_1 = x_3 = p$ and $x_2 = x_4 = q$ with $p, q \in \{0, 1\}$;
		$w = 2(p + q)$. The four sub-classes are listed below; in each
		row the right-most column collapses via~\eqref{eq:qubit-t-m4-trig}.
		\begin{center}
			\small
			\renewcommand{\arraystretch}{1.2}
			\begin{tabular}{c|ccc|c}
				\toprule
				$(p, q)$ & $[\sigma_1]_x$ & $[\sigma_2]_x$ & $[\sigma_3]_x$ & $\sum_j c_j [\sigma_j]_x$ \\
				\midrule
				$(0, 0)$ & $\frac{1}{2\sqrt 2}$  & $\frac{1}{4}$  & $\frac{1}{2\sqrt 2}$  & $\frac{2 \cos(\pi/12)}{3\sqrt 2} + \frac{1}{6} = \frac{2 + \sqrt 3}{6}$ \\
				$(0, 1)$ & $-\frac{1}{2\sqrt 2}$ & $\frac{1}{4}$  & $\frac{i}{2\sqrt 2}$  & $\frac{(\cos - \sin)(\pi/12)}{3\sqrt 2}(i - 1) + \frac{1}{6} = \frac{i}{6}$ \\
				$(1, 0)$ & $\frac{i}{2\sqrt 2}$  & $-\frac{1}{4}$ & $\frac{1}{2\sqrt 2}$  & $\frac{(\cos - \sin)(\pi/12)}{3\sqrt 2}(1 + i) - \frac{1}{6} = \frac{i}{6}$ \\
				$(1, 1)$ & $-\frac{i}{2\sqrt 2}$ & $-\frac{1}{4}$ & $\frac{i}{2\sqrt 2}$  & $\frac{2 \sin(\pi/12)}{3\sqrt 2} - \frac{1}{6} = \frac{\sqrt 3 - 2}{6}$ \\
				\bottomrule
			\end{tabular}
		\end{center}
		Each row matches the $w \in \{0, 2, 2, 4\}$ target.

		\paragraph{Case $(a, b) = (1, 0)$.}
		Only $\sigma_1$ and $\sigma_2$ contribute. Writing
		$x_1 = x_3 = p$ and $x_2 \ne x_4$ gives
		$[\sigma_1]_x = i^p / (2\sqrt 2)$ and
		$[\sigma_2]_x = i (-1)^p / 4$, so
		\begin{align}
			p = 0 (w = 1):\quad &
			[c_1 \sigma_1 + c_2 \sigma_2]_x
			= \frac{e^{i \pi / 12}}{3\sqrt 2} + \frac{i}{6}
			= \frac{(\sqrt 3 + 1)(1 + i)}{12},\\
			p = 1 (w = 3):\quad &
			[c_1 \sigma_1 + c_2 \sigma_2]_x
			= \frac{i e^{i \pi / 12}}{3\sqrt 2} - \frac{i}{6}
			= \frac{(\sqrt 3 - 1)(-1 + i)}{12},
		\end{align}
		using~\eqref{eq:qubit-t-m4-prefactor} in each line; these match
		the $w = 1$ and $w = 3$ targets.

		\paragraph{Case $(a, b) = (0, 1)$.}
		Only $\sigma_2$ and $\sigma_3$ contribute. Writing
		$x_2 = x_4 = q$ and $x_1 \ne x_3$ gives
		$i^{x_2 + x_4} = i^{2q} = (-1)^q$,
		$(-1)^{x_1 x_3 + x_2 x_4} = (-1)^{0 + q} = (-1)^q$ (since
		$x_1 \ne x_3 \in \{0, 1\}$ forces $x_1 x_3 = 0$), so
		$[\sigma_2]_x = (-1)^q (-1)^q / 4 = 1/4$ and
		$[\sigma_3]_x = i^{1 + q} / (2\sqrt 2)$, so
		\begin{align}
			q = 0 (w = 1):\quad &
			[c_2 \sigma_2 + c_3 \sigma_3]_x
			= \frac{1}{6} + \frac{i e^{-i \pi / 12}}{3\sqrt 2}
			= \frac{(\sqrt 3 + 1)(1 + i)}{12},\\
			q = 1 (w = 3):\quad &
			[c_2 \sigma_2 + c_3 \sigma_3]_x
			= \frac{1}{6} - \frac{e^{-i \pi / 12}}{3\sqrt 2}
			= \frac{(\sqrt 3 - 1)(-1 + i)}{12},
		\end{align}
		using the conjugate
		$e^{-i \pi / 12} / (3\sqrt 2) = (\sqrt 3 + 1)/12 - i (\sqrt 3 - 1)/12$
		of~\eqref{eq:qubit-t-m4-prefactor}; again matching the
		$w = 1$ and $w = 3$ targets.

		The four cases cover $\F{2}^4$, so~\eqref{eq:qubit-t-m4-id}
		holds pointwise. The exponent bound is
		then~\eqref{eq:exponent-from-finite}.
	\end{proof}

	\textit{Machine-checked Lean formalization:}
	\texttt{QubitTM4.lean}.
	
\end{document}